\documentclass[acmsmall,screen,nonacm]{acmart}
\usepackage{amsmath}
\usepackage{booktabs}   
\usepackage{subcaption}
\usepackage{enumerate}
\usepackage{tikz}
\usetikzlibrary{backgrounds, bending}
\usetikzlibrary{shapes.geometric}
\usepackage{listings}
\usepackage{csvsimple} 
\usepackage{multirow}
\usepackage{amsfonts}
\usepackage{url}
\usepackage{babel}
\usepackage{mathrsfs}
\usepackage[strings]{underscore}
\usepackage{enumerate,paralist}
\usepackage{lstautogobble}
\usepackage{makecell}
\usepackage{pifont}
\usepackage{paralist}
\usepackage{hyperref}
\usepackage{microtype}
\usepackage{float}
\usepackage{enumitem}
\usepackage{hhline}
\usepackage{tabularx}
\usepackage[table]{xcolor}
\usepackage{diagbox}
\usepackage[T1]{fontenc}
\usepackage{graphicx}

\newtheorem*{theorem*}{Theorem}
\newtheorem*{lemma*}{Lemma}
\newtheorem{remark}{Remark}
\newtheorem{example}{Example}
\newtheorem*{example*}{Example}
\AtBeginDocument{%
  }

\setcopyright{acmlicensed}
\copyrightyear{2018}
\acmYear{2018}
\acmDOI{XXXXXXX.XXXXXXX}

\acmConference[PLDI '26]{June 15--19,
  2026}{Boulder, Colorado}

\newcommand{\eps}{\epsilon}
\newcommand{\R}{\mathbb{R}}

\newcommand{\N}{\mathbb{N}}
\newcommand{\I}{\mathbb{I}}

\newcommand{\support}{\mathit{supp}}

\newcommand{\pCFG}{\mathcal{C}}
\newcommand{\locinit}{\loc_{\mathit{in}}}
\newcommand{\locterm}{\loc_{\mathit{out}}}

\newcommand{\vecinit}{\mathbf{x}_{\mathit{in}}}
\newcommand{\thetainit}{\theta_{\mathit{in}}}

\newcommand{\updates}{\mathit{Up}}

\newcommand{\Run}{\mathit{Run}}

\newcommand{\transitions}{\mapsto}

\newcommand{\guards}{G}
\newcommand{\prob}{\mathit{Pr}}

\newcommand{\pvars}{\mathcal{V}}
\newcommand{\locs}{\mathit{L}}
\newcommand{\loc}{\ell}

\newcommand{\probm}{\mathbb{P}}
\newcommand{\E}{\mathbb{E}}

\newcommand{\val}{\mathbf{x}}

\newcommand{\timeterm}{\mathit{TimeTerm}}
\newcommand{\Termset}{\mathit{Term}}
\newcommand{\Output}{\mathit{Output}}

\newcommand{\out}{\mathit{out}}

\newcommand{\pcfg}{\pCFG}

\newcommand{\similar}{\sim_{\Phi}}

\renewcommand{\H}{\mathbb{H}}
\renewcommand{\L}{\mathcal{L}}
\newcommand{\U}{\mathcal{U}}
\newcommand{\distrs}{\mathcal{D}}
\newcommand{\Vout}{\pvars_{\mathit{out}}}

\newcommand{\parag}[1]{\smallskip\indent{\em {#1}}}
\begin{document}

\title{SuperDP: Differential Privacy Refutation via Supermartingales}

\author{Krishnendu Chatterjee}
\email{Krishnendu.Chatterjee@ist.ac.at}
\affiliation{
\institution{Institute of Science and Technology Austria (ISTA)}
\city{Klosterneuburg}
\country{Austria}
}
\author{Ehsan Kafshdar Goharshady}
\email{ehsan.goharshady@ist.ac.at}
\affiliation{
\institution{Institute of Science and Technology Austria (ISTA)}
\city{Klosterneuburg}
\country{Austria}
}
\author{\DJ or\dj e \'Zikeli\'c}
\email{dzikelic@smu.edu.sg}
\affiliation{
\institution{Singapore Management University (SMU)}
\city{Singapore}
\country{Singapore}
}


\begin{abstract}
  Differential privacy (DP) has established itself as one of the standards for ensuring privacy of individual data. However, reasoning about DP is a challenging and error-prone task, hence methods for formal verification and refutation of DP properties have received significant interest in recent years. In this work, we present a novel method for automated formal refutation of $\epsilon$-DP. Our method refutes $\epsilon$-DP by searching for a pair of inputs together with a non-negative function over outputs whose expected value on these two inputs differs by a significant amount. The two inputs and the non-negative function over outputs are computed simultaneously, by utilizing upper expectation supermartingales and lower expectation submartingales from probabilistic program analysis, which we leverage to introduce a sound and complete proof rule for $\epsilon$-DP refutation. To the best of our knowledge, our method is the first method for $\epsilon$-DP refutation to offer the following four desirable features: (1)~it is fully automated, (2)~it is applicable to stochastic mechanisms with sampling instructions from both discrete and continuous distributions, (3)~it provides soundness guarantees, and (4)~it provides semi-completeness guarantees. Our experiments show that our prototype tool SuperDP achieves superior performance compared to the state of the art and manages to refute $\epsilon$-DP for a number of challenging examples collected from the literature, including ones that were out of the reach of prior methods.
\end{abstract}

\maketitle

\section{Introduction}\label{sec:intro}


\parag{Differential privacy.} Differential privacy is a formal notion of privacy that tackles the problem of extracting statistical information about a dataset, without compromising the privacy of individual data subjects~\cite{Dwork06,DworkR14}. It has gained great prominence over the recent years as one of the standards for ensuring privacy in (randomized) computation~\cite{BittauEMMRLRKTS17,DingKY17,JohnsonNS18}. Intuitively, a randomized algorithm or a stochastic mechanism is said to be differentially private, if slightly changing the input of the algorithm (e.g.~changing the data of a single user) does not significantly affect the probability distribution of outputs. Formally, given a probabilistic program $P$ modeling the randomized computation, a similarity relation $\sim_\Phi$ over its inputs, and a privacy budget $\epsilon \geq 0$, $P$ is said to be {\em $\epsilon$-differentially private ($\epsilon$-DP)} if for every pair of similar inputs $\vecinit^1 \sim_\Phi \vecinit^2$ the resulting output distributions of $P$ are not too distant from each other, i.e.~the inequality
\begin{equation}\label{eq:cond}
    \mathbb{P}_{\vecinit^1}\Big[\Output(A)\Big] \leq e^\eps \cdot \mathbb{P}_{\vecinit^2}\Big[\Output(A)\Big],
\end{equation}
holds for every set of outputs $A$ in $P$. Smaller values of $\epsilon$ provide stronger privacy guarantees, with $\epsilon = 0$ ensuring perfect privacy and the indistinguishability of computation due to similar inputs.

\parag{Correctness of differential privacy.} While DP provides an intuitive and mathematically clean definition of privacy, it turns out that reasoning about it is a difficult task. For instance, minor tweaks to correct mechanisms can break differential privacy~\cite{DworkNRRV09,LyuSL17}. Moreover, as any non-trivial property of (probabilistic) programs, the problem of DP checking for Turing-complete probabilistic programming languages is undecidable~\cite{BartheCJS020}. 
This inherent difficulty, together with the importance and wide adoption of differential privacy as one of the standards for private computation, lead to a significant interest and research effort towards enabling automated reasoning about DP.

Automated reasoning about DP can be divided into two complementary problems -- (1)~{\em DP verification}, i.e.~the problem of proving that a stochastic mechanism is differentially private, and (2)~{\em DP refutation}, i.e.~the problem of proving that a stochastic mechanism is {\em not} differentially private. Due to the undecidability of DP checking, one cannot hope for a single algorithm that can verify and refute DP for all stochastic mechanisms, hence these two problems require different approaches.

The DP verification problem has received significant attention, see e.g.~\cite{BartheCJS020,BunGG22,BartheKOB13,BartheGGHS16,WangDKZ20,AvanziniBDG25}, and we discuss existing methods and approaches in Section~\ref{sec:relatedwork}. The DP refutation problem, on the other hand,  has received comparatively less attention. Existing automated methods for DP refutation can be classified into static and dynamic methods. Dynamic methods aim to refute DP by sampling program executions and deriving statistical lower bounds on distances between output distributions on different inputs. If these distances are sufficiently significant, one can conclude that the mechanism is not DP with statistical guarantees~\cite{DingWWZK18}. Static methods, on the other hand, aim for formal guarantees by performing static analysis of probabilistic programs that implement the stochastic mechanism of interest. However, most existing static analyses for DP refutation are either not fully automated, or are restricted to finite or countable programs that do not allow real-valued program variables or sampling from continuous distributions~\cite{FarinaCG21,BartheCJS020}, which is a significant limitation given that continuous distributions such as Laplacians are standard in many DP mechanisms used in practice~\cite{Dwork06,DworkR14,opendp}. A notable exception is CheckDP~\cite{WangDKZ20} (successor of LigthDP~\cite{ZhangK17} and ShadowDP~\cite{WangDWKZ19}), which searches for two program inputs whose respective executions cannot be aligned to yield the same output without changing the program by much. They then use the off-the-shelf exact inference tool PSI~\cite{GehrMV16} to formally prove that $\epsilon$-DP is violated for the identified pair of inputs.


\parag{Limitations of previous approaches.} While the above works present significant advances in automated reasoning about DP, the existing support for automated and formal DP refutation is still limited and DP refutation remains a challenging problem. In particular, to the best of our knowledge, no existing method for DP refutation provides all of the following desirable features:
\begin{compactenum}
    \item {\em Automation.} We are interested in fully automated methods for DP refutation.
    \item {\em Support for general probabilistic programs.} We are interested in methods that are applicable to stochastic mechanisms that allow sampling from both discrete and continuous distributions. As discussed above, this is highly important because continuous distributions, such as Laplacians, are standard in many classical DP mechanisms~\cite{Dwork06,DworkR14,opendp}.
    \item {\em Soundness guarantees.} We are interested in static analysis methods that can provide formal guarantees on the correctness of DP refutations that they report.
    \item {\em Semi-completeness guarantees.} Given the undecidable nature of DP checking, one cannot hope for an algorithm which is sound and complete. However, as is the case in many static program analyses such as termination~\cite{ColonS01,PodelskiR04,PodelskiR04-2} or safety~\cite{ColonSS03,ChatterjeeFGG2020} proving, it can be beneficial to provide semi-completeness guarantees, i.e.~a set of conditions under which a method is guaranteed to refute DP. This is particularly important in the context of DP analysis, given how sensitive DP mechansisms are to even minor tweaks~\cite{DworkNRRV09,LyuSL17}. CheckDP~\cite{WangDKZ20}, which is the only method that satisfies the above three features, does not provide any form of semi-completeness guarantees.
\end{compactenum}

\parag{Our contributions.} In this work, we propose a method for $\epsilon$-DP refutation that provides all four of the desirable features discussed above. The key idea behind our method is to search for a pair of similar inputs $\vecinit^1 \sim_\Phi \vecinit^2$ for which $\epsilon$-DP is violated by showing that they lead to "expectation mismatch". That is, rather than explicitly searching for an output event $A$ for which eq.~\eqref{eq:cond} is violated, our method instead searches for a non-negative function $f$ over the outputs of the stochastic mechanism for which
\begin{equation}\label{eq:cond2}
    \mathbb{E}_{\vecinit^1}[f] > e^\epsilon \cdot \mathbb{E}_{\vecinit^2}[f].
\end{equation}
We prove that this is both a necessary and sufficient condition for the pair of inputs to violate $\epsilon$-DP.

A challenge in computing such a function $f$ is that the exact computation of the expected value of $f$ on output is in general infeasible, as it may not always admit a closed-form expression~\cite{KaminskiKMO18}. To overcome this challenge, our method instead computes a pair of similar inputs $\vecinit^1 \sim_\Phi \vecinit^2$ and a function $f$ {\em simultaneously and together with} a lower bound $L \leq \mathbb{E}_{\vecinit^1}[f]$ and an upper bound $U \geq \mathbb{E}_{\vecinit^2}[f]$ on the two expected values, while in addition enforcing $L > e^{\epsilon} \cdot U$. These three inequalities together entail the inequality in eq.~\eqref{eq:cond2} and hence refute $\epsilon$-DP.

To reason about and automatically compute upper and lower bounds on the expected value of $f$ on output, we leverage the notions of upper expectation supermartingales (UESMs) and lower expectation submartingales (LESMs), which were shown to provide sound and complete proof rules for computing upper and lower bounds on the expected value of a function on output~\cite{ChatterjeeGNZ24,ChatterjeeGNZ25}. This allows us to formulate what we call {\em expectation super/sub-martingale witness for $\epsilon$-DP refutation}, which are tuples $(\vecinit^1,\vecinit^2,f,L_f,U_f)$ that consist of a pair of similar inputs, a non-negative function $f$ over outputs, and a U/LESM pair which give rise to expectation bounds that together imply eq.~\eqref{eq:cond2}. We show that our expectation super/sub-martingale witness for $\epsilon$-DP refutation provide a {\em sound and complete proof rule} for $\epsilon$-DP refutation in stochastic mechanisms that can be modeled as probabilistic programs that allow general Borel-measurable arithmetic expressions and sampling instructions from both discrete and continuous probability distributions. We then use this sound and complete proof rule to design a {\em sound and conditionally semi-complete algorithm} for $\epsilon$-DP refutation. Our algorithm is applicable to stochastic mechanisms that can be modeled as polynomial arithmetic probabilistic programs (while allowing sampling instructions from distributions whose probability density function is not polynomial, e.g.~Laplacian or normal distributions), and it follows a template-based synthesis approach to compute an instance of an expectation super/sub-martingale witness for $\epsilon$-DP refutation which can be specified in the real-valued polynomial arithmetic. Finally, we develop and experimentally evaluate a prototype tool SuperDP ({\bf Super}martingale-based {\bf D}ifferential {\bf P}rivacy refutation), and our experiments demonstrate our method's practical applicability and favourable performance compared to the state of the art on automated $\epsilon$-DP refutation.

Our contributions can be summarized as follows:
\begin{compactenum}
    \item {\em Theory: Sound and complete proof rule for $\epsilon$-DP refutation.} We introduce {\em expectation super/sub-martingale witness for $\epsilon$-DP refutation}, a sound and complete proof rule for $\epsilon$-DP refutation in stochastic mechanisms that can be modeled as probabilistic programs that allow general Borel-measurable arithmetic expressions and sampling instructions from both discrete and continuous probability distributions.
    \item {\em Automation: Sound and conditionally semi-complete algorithm for $\epsilon$-DP refutation.} We design a sound and conditionally semi-complete algorithm for $\epsilon$-DP refutation in stochastic mechanisms that can be modeled as polynomial arithmetic probabilistic programs.
    \item {\em Experimental evaluation.} Our experiments show that our sound and conditionally semi-complete algorithm is able to refute $\epsilon$-DP in a number of challenging examples collected form the literature, including examples that no prior automated method could handle.
\end{compactenum}

\section{Overview}\label{sec:overview}

To illustrate the key ideas behind our method for refuting DP, we will consider the randomized response mechanism in Fig.~\ref{fig:RR-1} and the histogram mechanism in Fig.~\ref{fig:histogram}, both of which are basic and fundamental mechanisms in the DP literature~\cite{DworkR14}. In both cases, we use probabilistic programs (PPs) as a language for specifying stochastic mechanisms whose DP we wish to analyze. PPs are classical imperative or functional programs extended with the ability to sample values from probability distributions and to assign sampled values to program variables, and they provide a general expressive framework for modeling and specifying stochastic models and protocols~\cite{Ghahramani15,MeentPYW18}.

\begin{figure}[t]
	\centering
	\begin{subfigure}[t]{0.47\textwidth}
		\begin{lstlisting}[frame=single,numbers=none]
Input: $x \in \{0,1\} \wedge out = 0$
Sim: $(out^1 = out^2 = 0)$
$l_1$:  if prob(0.5) then
$l_2$:    $out$ := $x$
   else
$l_3$:    if prob(0.5) then
$l_4$:      $out$ := 0
     else
$l_5$:      $out$ := 1
$l_t$:
Output: $out$
		\end{lstlisting}
        \caption{Randomized response mechanism (\texttt{RR-1})~\cite{DworkR14}: provides plausible deniability by ensuring $(\ln{3})$-DP.}
        \label{fig:RR-1}
	\end{subfigure}%
\hfill
\begin{subfigure}[t]{0.47\textwidth}
		\begin{lstlisting}[frame=single,numbers=none]
Input: $q \in \R^n$
Sim: ($\eta^1 = \eta^2 = 0) \wedge (\forall i. \, |q^1_i - q^2_i| \leq 1$)
      $ \wedge \, (\forall i. 
    \, q^1_i \neq q^2_i \Rightarrow \forall j\neq i. \, q^1_j= q^2_j)$
$l_1$:  for $i \in [0,n)$:
$l_2$:    $\eta$ := laplace($1$)
$l_3$:    $q_i$ := $q_i + \eta$
$l_t$:
            \end{lstlisting}
        \begin{lstlisting}[frame=single,numbers=none]
$l_1$: $\eta$ := laplace($1$)
$l_2$: $q_0$ := $q_0 + \eta$
$l_t$:
Output: $q$
            \end{lstlisting}
        \caption{Histogram mechanism \cite{DworkR14} (top) and its unrolling for $n=1$ (bottom): given true query answers $q$, adds independent Laplacian perturbation to each element to satisfy $1$-DP.}
        \label{fig:histogram}
	\end{subfigure}%
        \caption{Our two running examples for illustrating our approach to refuting $\epsilon$-DP. In both cases, we use \textsc{Input} to denote the predicate that defines possible program inputs, \textsc{Sim} to denote the similarity relation over inputs, and $\textsc{Out}$ to denote output program variables.}
        \label{example_ts}
\end{figure}

\begin{example}[Running example: Randomized response mechanism]
    The PP in Fig.~\ref{fig:RR-1} implements the well-known "randomized response" mechanism~\cite{DworkR14}. This mechanism allows individuals to respond to sensitive queries while providing them with plausible deniability. Suppose that input $x$ is sensitive information, e.g. whether or not a surveyed student has cheated in an exam. The student is asked to toss a fair coin (line $l_1$), answer truthfully if it lands on heads (line $l_2$), and otherwise answer uniformly at random by throwing another fair coin (lines $l_3$ to $l_5$). This way, even if the student responds 'yes', they can credibly deny their cheating by claiming that the first coin landed on tails. 
\end{example}

\begin{example}[Running example: Histogram mechanism]
    The PP in Fig.~\ref{fig:histogram} implements the histogram mechanism introduced in~\cite{DworkR14}. The input to this mechanism is a vector of true answers to some number of queries $q$ about a database (e.g.~a query $q_i$ may ask for the number of database elements that are equal to some given value). The goal of the histogram mechanism is to output the query answers privately by adding noise to them, where the noise is sampled according to the Laplacian distribution with mean $0$ and scale $\frac{1}{\epsilon}$. Here, we consider $\epsilon=1$ for illustration.
\end{example}

\parag{Differential privacy refutation.} Let $P$ be a PP that models a stochastic mechanism whose differential privacy we wish to analyze. Intuitively, we say that $P$ is DP, if for any pair of inputs to $P$ that are sufficiently "similar", the probability of every output event on these two inputs is also similar. 

This intuition is formalized as follows. Let $\pvars$ be the set of program variables in $P$ with the set of output variables $\Vout \subseteq \pvars$, $\thetainit \subseteq \R^{|\pvars|}$ be the set of all inputs to $P$, and $\Phi \subseteq \thetainit \times \thetainit$ be a similarity relation over the inputs that specifies which input pairs are deemed to be "similar". For a pair of inputs $\vecinit^1,\vecinit^2\in\thetainit$, we write $\vecinit^1 \sim_\Phi \vecinit^2$ to denote that $(\vecinit^1,\vecinit^2) \in \Phi$. Then, for $\epsilon \geq 0$, we say that the PP $P$ is {\em $\epsilon$-differentially private ($\epsilon$-DP)}, if for any pair of similar inputs $\vecinit^1 \sim_\Phi \vecinit^2$ and for any event $A$ over the PP outputs, we have
\begin{equation*}
    \mathbb{P}_{\vecinit^1}\Big[\Output(A)\Big] \leq e^\epsilon \cdot \mathbb{P}_{\vecinit^2}\Big[\Output(A)\Big],
\end{equation*}
where we use $\mathbb{P}_{\vecinit^1}$ and $\mathbb{P}_{\vecinit^2}$ to denote the probability measures induced by the PP on these two inputs. In words, the probability of any output event on two similar inputs may differ by a multiplicative factor of at most $e^\epsilon$. Given a PP $P$, a similarity relation $\Phi$ over its inputs and a privacy budget $\epsilon \geq 0$, the DP refutation problem is concerned with formally proving that $P$ is not $\epsilon$-DP.

\begin{example}\label{ex:dpoverview}
    The randomized response mechanism in Fig.~\ref{fig:RR-1} is $\epsilon$-DP for $\epsilon = \ln 3$, but not $\epsilon$-DP for $\epsilon < \ln 3$~\cite{DworkR14}.
    Indeed, Table~\ref{tab:dist-RR-1} shows the output distribution of the randomized response mechanism for each possible input pair $(x,\out)$. Based on the data in the table, it can be verified by inspection that $\mathbb{P}_{\vecinit^1}[\Output(A)] \leq e^{\ln 3} \cdot \mathbb{P}_{\vecinit^2}[\Output(A)]$ holds for all similar inputs $\vecinit^1 = (x^1,\out^1), \vecinit^2 = (x^2,\out^2) \in \{0,1\}^2$ and for all output events $A$. Hence, the mechanism is $(\ln 3)$-DP. However, the mechanism is {\em not} $\epsilon$-DP for $\epsilon < \ln 3$. This is because, for the pair of similar inputs $(x^1,\out^1) = (1,0)$ and $(x^2,\out^2) = (0,0)$ and the output event $A = (\out = 1)$, we have $\mathbb{P}_{\vecinit^1}[\Output(A)] = e^{\ln 3} \cdot \mathbb{P}_{\vecinit^2}[\Output(A)]$ so the inequality is tight for $\epsilon = \ln 3$ and would be violated for $\epsilon < \ln 3$.
    
    \begin{table}[t]
    \begin{center}   
    \begin{tabular}[h]{c|c|c}
         \diagbox[width=1.1cm, height=0.5cm]{\scriptsize in}{\scriptsize out}& $0$ & $1$ \\
         \hline
         $(0,0)$ & $0.75$ & $0.25$ \rule{0pt}{10pt} \\
         \hline 
         $(1,0)$ & $0.25$ & $0.75$ \rule{0pt}{10pt} \\
    \end{tabular}
    \end{center}
    \caption{Output distribution of the randomized response mechanism in Fig. \ref{fig:RR-1}.}
    \label{tab:dist-RR-1}
    \end{table}

    On the other hand, it is a classical result in the literature on differential privacy that the histogram mechanism in Fig.~\ref{fig:histogram} is $1$-DP but not $\epsilon$-DP for $\epsilon < 1$ \cite{Dwork06}. Intuitively, each query answer is perturbed independently by Laplace noise with scale $1$ to mask the query answer, so the mechanism is $1$-DP but not $\epsilon$-DP for any $\epsilon < 1$.
\end{example}

\parag{Our approach: DP refutation via expectation mismatch.} By the above definition, to refute $\epsilon$-DP of $P$, it suffices to find a pair of similar inputs $\vecinit^1 \sim_\Phi \vecinit^2$ and an output event $A$ such that $\mathbb{P}_{\vecinit^1}[\Output(A)] > e^\epsilon \cdot \mathbb{P}_{\vecinit^2}[\Output(A)]$. Our method for DP refutation explicitly computes such a pair of similar inputs $\vecinit^1 \sim_\Phi \vecinit^2$. However, rather than also searching for an output event $A$, it instead computes a non-negative function $f: \R^{|\Vout|} \rightarrow \R$ over PP outputs for which
\begin{equation}\label{eq:expectrefute}
    \mathbb{E}_{\vecinit^1}[f] > e^\epsilon \cdot \mathbb{E}_{\vecinit^2}[f].
\end{equation}
Here, we use $\mathbb{E}_{\vecinit^1}$ and $\mathbb{E}_{\vecinit^2}$ to denote the expectation operators over the output distributions induced by the PP on these two inputs. We show in the proof of Theorem~\ref{thm:ref-sound-complete} that $P$ is not $\epsilon$-DP if and only if there exist a pair of similar inputs $\vecinit^1 \sim_\Phi \vecinit^2$ and a non-negative function $f$ over outputs for which eq.~\eqref{eq:expectrefute} holds. Hence, in order to refute $\epsilon$-DP of the PP $P$, we may without loss of generality search for a pair of similar inputs $\vecinit^1 \sim_\Phi \vecinit^2$ and such a non-negative function $f$ over outputs.

\begin{example}
    Consider again the randomized response mechanism in Fig. \ref{fig:RR-1}. Let $f(\out) = \out^2$ be a non-negative function over its output variable. Based on the probabilities in Table \ref{tab:dist-RR-1}, one can observe that for two similar inputs $(x^1_\mathit{in},\out^1_\mathit{in}) = (1,0)$ and $(x^2_\mathit{in},\out^2_\mathit{in}) = (0,0)$, we have
    \[
         \mathbb{E}_{(1,0)}[f] = 0.75, \,\,\,\,\,\,\,\,\,\,\,\,\,\,  \mathbb{E}_{(0,0)}[f] = 0.25.
    \]
    Hence, we have $\mathbb{E}_{(1,0)}[f] = e^{\ln 3} \cdot \mathbb{E}_{(0,0)}[f]$, and
    for any $\epsilon < \ln 3$ it holds that $\mathbb{E}_{(1,0)}[f] > e^\epsilon \cdot \mathbb{E}_{(0,0)}[f]$. This refutes $\epsilon$-DP of the mechanism for every $\epsilon < \ln 3$. 

    Next, consider the histogram mechanism in Fig. \ref{fig:histogram}. Already in the case when $n=1$ (the simplified PP for $n=1$ is showed at the bottom of Fig. \ref{fig:histogram}), if we consider two similar inputs $(q_{0\,\mathit{in}}^1,\eta_\mathit{in}^1) = (1.5,0)$ and $(q_{0\, \mathit{in}}^2,\eta_\mathit{in}^2) = (2.5,0)$ and a non-negative function $f(q_0) = 2q_0^2 + 393.92 q_0^4$ over the output variable~$q_0$, we observe that
    $\mathbb{E}_{(1.5,0)}[f] = 22092.64 \leq e^{\ln 2.46} \cdot 54402.08 = e^{\ln 2.46} 
    \cdot \mathbb{E}_{(2.5,0)}[f]$, which implies that the mechanism is not $\epsilon$-DP for any $\epsilon < \ln 2.46 \approx 0.9001$. 

\end{example}

\parag{Sound and complete proof rule for DP refutation: Expectation super/sub-martingale witnesses.} A challenge in computing such a function $f$ is the exact computation of the expected value of $f$ on PP output, since this expected value may not admit a closed-form expression in the general case (e.g.~for PPs with unbounded loops~\cite{KaminskiKMO18}). To overcome this challenge, our method instead computes a non-negative function $f$ over PP outputs {\em simultaneously and together with} a lower bound $L \leq \mathbb{E}_{\vecinit^1}[f]$ and an upper bound $U \geq \mathbb{E}_{\vecinit^2}[f]$ on the two expected values, while in addition requiring that $L > e^{\epsilon} \cdot U$. These three inequalities together entail the inequality in eq.~\eqref{eq:expectrefute} and hence allow our method to refute $\epsilon$-DP.

Unlike the exact expected value of the function $f$ on PP output which is in general hard to compute, upper and lower bounds on the expected value may be efficiently and automatically computed by using {\em upper expectation supermartingales (UESMs)} for upper bounds and {\em lower expectation submartingales (LESMs)} for lower bounds~\cite{ChatterjeeGNZ24,ChatterjeeGNZ25}. A UESM (resp.~LESM) for a function $f$ over PP outputs is a function that assign a real value to each PP state, which is required to satisfy a set of conditions that together ensure that the UESM (resp.~LESM) evaluates to an upper (resp.~lower) bound on the expected value of $f$ on PP output. It was shown in~\cite{ChatterjeeGNZ25} that U/LESMs provide a sound and complete proof rule for deriving upper and lower bounds on the expected value of a function on PP output. Their name is due to their connection to super- and submartignale processes from probability theory, which are used to formally establish soundness and completeness results~\cite{Williams91}. Our method leverages U/LESMs and we propose a new proof rule for refuting $\epsilon$-DP in a PP which we call expectation super/sub-martingale witnesses. Informally, given a PP, a similarity relation $\Phi$ over its inputs and a privacy budget $\epsilon \geq 0$, an {\em expectation super/sub-martingale witness for $\epsilon$-DP refutation} in the PP is a tuple $(\vecinit^1,\vecinit^2,f,L_f,U_f)$ where:
\begin{compactenum}
    \item[(R1)] $\vecinit^1, \vecinit^2$ is a pair of similar inputs, i.e.~$\vecinit^1 \similar \vecinit^2$;
    \item[(R2)] $f\colon\R^{\pvars_\out} \to \R$ is a non-negative function over PP outputs;
    \item[(R3)] $L_f$ is an LESM for $f$ that satisfies the necessary LESM conditions and yields a lower expected value bound $L \leq \mathbb{E}_{\vecinit^1}[f]$;
    \item[(R4)] $U_f$ is a UESM for $f$ that satisfies the necessary UESM conditions and yields an upper expected value bound $U \geq \mathbb{E}_{\vecinit^2}[f]$
    \item[(R5)] The inequality $L > e^\epsilon \cdot U$ holds.
\end{compactenum}
We formalize the notions of U/LESMs, the conditions that they need to satisfy, and our novel proof rule for $\epsilon$-DP refutation in Section~\ref{sec:proofrule}, and also illustrate it on our two running examples in Fig.~\ref{fig:RR-1} and Fig.~\ref{fig:histogram}. Moreover, in Theorem~\ref{thm:ref-sound-complete}, we prove that expectation super/sub-martingale witness for $\epsilon$-DP refutation provide a {\em sound and complete} proof rule for refuting $\epsilon$-DP in PPs. This is the main theoretical result of our work.

\parag{Automation: Template-based synthesis.} We also present an algorithm for automated $\epsilon$-DP refutation based on our novel sound and complete proof rule. Our algorithm considers arithmetic PPs in which all arithmetic expressions are assumed to be polynomials (see the full list of algorithm assumptions in Section~\ref{sec:algo}) and it computes an instance of an expectation super/sub-martingale witness for $\epsilon$-DP refutation that can be expressed in polynomial real arithmetic.

The key challenge in achieving this is the effective computation of all elements of the witness tuple $(\vecinit^1,\vecinit^2,f,L_f,U_f)$. Note that the elements of the tuple cannot be computed separately, and the computation must be guided by the requirement that the LESM $L_f$ and the UESM $U_f$ give rise to lower and upper bounds $L$ and $U$ on the expected value of function $f$ on output that satisfy $L > e^{\epsilon} \cdot U$. Hence, we propose an algorithm that computes these five objects simultaneously while also imposing this inequality on the computation process, by following a template-based synthesis approach. Our algorithm fixes a symbolic template for each of the five elements of the witness tuple $(\vecinit^1,\vecinit^2,f,L_f,U_f)$, where templates for $f$, $L_f$ and $U_f$ are defined in terms of polynomial expressions over real-valued symbolic template variables. It then collects a system of constraints over the symbolic template variables which together entail that $(\vecinit^1,\vecinit^2,f,L_f,U_f)$ is a valid expectation super/sub-martingale witness for $\epsilon$-DP refutation. The collected system of constraints is processed by applying results from real algebraic geometry, in order to simplify its solving. Finally, the algorithm uses an off-the-shelf SMT solver to solve the system of constraints, and any solution gives rise to a valid proof of $\epsilon$-DP refutation. We present the details of our template-based synthesis algorithm in Section~\ref{sec:algo}. Moreover, in Theorem~\ref{thm:algo-sound-complete}, we show that the algorithm is sound, conditionally semi-complete (a notion that we formalize in Section~\ref{sec:algo}), and admits a PSPACE upper bound on its complexity.

\parag{Novelty and limitations.} As discussed in Section~\ref{sec:intro}, our theoretical results give rise to a sound and complete proof rule for $\epsilon$-DP refutation that allows fully automated reasoning. Our experiments in Section~\ref{sec:experiments} also demonstrate the ability of our automated method to refute $\epsilon$-DP for a large number of challenging stochastic mechanisms and superior practical performance compared to the state of the art. This is the main novelty of our work. However, our method still suffers from a few limitations which we discuss below, addressing which is an exciting venue for future work:
\begin{compactenum}
    \item {\em Theory: Restriction to $\epsilon$-DP refutation.} Our method considers an established notion of $\epsilon$-DP and the problem of its refutation. However, another popular notion of DP is $(\epsilon,\delta)$-DP~\cite{DworkR14}. Given a PP, a similarity relation $\Phi$ over its inputs and $\epsilon,\delta \geq 0$, a PP is said to be $(\epsilon,\delta)$-DP if for every pair of similar inputs $\vecinit^1 \sim_\Phi \vecinit^2$ and for every output event $A$ we have $\mathbb{P}_{\vecinit^1}[\Output(A)] \leq e^\epsilon \cdot \mathbb{P}_{\vecinit^2}[\Output(A)] + \delta$. Hence, $(\epsilon,\delta)$-DP is a more general notion of differential privacy with $\epsilon$-DP being the special case when $\delta = 0$.

    Our method for $\epsilon$-DP refutation {\em does not} readily extend to $(\epsilon,\delta)$-DP refutation. This is because the immediate extension of DP refutation via expectation mismatch reasoning as in eq.~\eqref{eq:expectrefute} becomes unsound in the case of $(\epsilon,\delta)$-DP refutation. Studying how to extend this style of reasoning to $(\epsilon,\delta)$-DP refutation and other notions of privacy is an interesting direction of future work. We also note that this restriction is not limited to our method, and most existing methods for formal and automated DP refutation are also restricted to $\epsilon$-DP refutation while not being applicable to $(\epsilon,\delta)$-DP refutation~\cite{WangDKZ20,BichselSBV21,BichselGDTV18,GuanFHY23}.
    \item {\em Automation: Restriction to polynomial arithmetic PPs.} Our novel proof rule and its soundness and completeness guarantees apply to general PPs in which arithmetic expressions may be arbitrary Borel-measurable functions. However, our template-based synthesis algorithm for automated computation of expectation super/sub-martingale witnesses is restricted to polynomial arithmetic PPs. Considering algorithms for DP refutation in non-polynomial PPs would be an interesting direction of future work. Note, however, that our class of polynomial arithmetic PPs still allows sampling instructions from non-polynomial distributions such as Laplacian or normal, and we only require that arithmetic expressions appearing in program variable assignments are polynomial expressions.
\end{compactenum}


\section{Preliminaries}\label{sec:prelims}

We use boldface notation to denote vectors, e.g.~$\mathbf{x}$ or $\mathbf{y}$. For a real-valued vector $\mathbf{x} \in \R^n$ and an index $1 \leq i \leq n$, we use $\mathbf{x}[i]$ to denote the value of the $i$-th component of $\mathbf{x}$. In what follows, we assume familiarity with basic notions of probability theory, such as {\em probability space}, {\em random variable}, or {\em expected value}. We use $\distrs(A)$ to define the set of all probability distributions over the set $A$. We refer the reader to~\cite{Williams91} for formal definitions of these notions.

\subsection{Probabilistic Programs}\label{sec:probprogs}

We formalize the problem of DP refutation in the setting of {\em probabilistic programs (PPs)}. We consider imperative arithmetic PPs that allow standard programming constructs, such as program variable assignments, if-branching, sequential composition and loops. In addition, we allow constructs for {\em sampling instructions} which sample a value from a probability distribution and assign it to some program variable. These are denoted by $\textbf{sample}(\dots)$ commands in our syntax. We allow sampling from both discrete and continuous probability distributions. Note that probabilistic branching instructions $\textbf{if prob}(\dots)$ can be obtained as syntactic sugar, by using a sampling instruction from a Bernoulli distribution followed by conditional branching. Example PPs that will serve as our running examples for DP analysis are shown in Fig.~\ref{fig:RR-1} and Fig.~\ref{fig:histogram}.

We assume that all program variables in our PPs are real-valued. Moreover, for the semantics of our PPs to be mathematically well defined, we assume that every arithmetic expression appearing in our PPs defines a {\em Borel-measurable} function over program variables. Borel-measurability is a standard assumption in the PP analysis literature, and allows most standard arithmetic operations and functions from mathematical analysis~\cite{Williams91}.

\parag{Probabilistic control flow graphs.} We use {\em probabilistic control-flow graphs (pCFGs)}~\cite{AgrawalC018,ChatterjeeFNH18,ChatterjeeGNZ24} to formally model our PPs.
A pCFG is a tuple $\pCFG=(\locs,\pvars,\Vout,\locinit,\thetainit,\transitions,\guards,\updates)$, where:
\begin{compactitem}
	\item $\locs$ is a finite set of {\em locations};
	\item $\pvars=\{x_1,\dots,x_{|\pvars|}\}$ is a finite set of {\em program variables};
	\item $\Vout = \{x_1,\dots,x_{|\Vout|}\} \subseteq \pvars$ is a finite set of {\em output variables};
	\item $\locinit \in \locs$ is the {\em initial program location};
    \item $\thetainit \subseteq \R^{|\pvars|}$ is the set of {\em initial variable valuations};
	\item $\transitions\,\subseteq \locs \times \distrs(\locs)$ is a finite set of {\em transitions}. For each transition $\tau=(\loc,\prob)$, $\loc$ is its {\em source location} and $\prob:\locs \rightarrow [0,1]$ is a probability distribution over {\em successor locations};
	\item $\guards$ is a map assigning to each transition $\tau=(\loc,\prob)\in\,\transitions$ a predicate $\guards(\tau)$ over $\pvars$ specifying if $\tau$ can be executed, which we call a {\em guard};
	\item $\updates$ is a map assigning to each transition $\tau=(\loc,\prob)\in\,\transitions$ an {\em update} $\updates(\tau)=(j,u)$, where $j\in\{1,\dots,|\pvars|\}$ is a {\em variable index} and $u$ is an {\em update element} which can be either
	\begin{compactitem}
		\item the bottom element $u=\bot$, denoting no variable update, or
		\item a Borel-measurable arithmetic expression $u:\R^{|\pvars|}\rightarrow\R$, or
		\item a probability distribution $u = \delta$, denoting that variable value is sampled according to $\delta$.
	\end{compactitem}
\end{compactitem}
The translation of PPs into pCFGs is standard, hence we omit the details and refer the reader to e.g.~\cite{AgrawalC018}. We make the following assumptions about our pCFGs:
\begin{compactenum}[(i)]
    \item A pCFG $\pCFG$ contains a special {\em terminal location} $\locterm$ which only has one outgoing self-loop transition $\tau$ with $\guards(\tau) \equiv \text{true}$ and the bottom update element. 
    \item Each location $\loc$ has at least one outgoing transition and $\bigvee_{\tau=(l,\_)}\guards(\tau)\equiv\mathit{true}$.
    \item The guards of any two distinct transitions $\tau_1$ and $\tau_2$ outgoing from the same location $\loc$ are {\em mutually exclusive}, i.e.~$\guards(\tau_1)\land\guards(\tau_2)\equiv\mathit{false}$
\end{compactenum}
The first two assumptions are to avoid deadlock situations where the program cannot progress, and are imposed without loss of generality as they can be enforced by introducing additional transitions. The last assumption prevents pCFGs from admitting non-deterministic behavior.

\parag{States, paths and runs.} A {\em state} in $\pCFG$ is a tuple $(\loc,\mathbf{x})$, where $\loc$ is a location and $\mathbf{x}\in\R^{|\pvars|}$ is a variable valuation. A transition $\tau=(\loc,\prob)$ is {\em enabled} at a state $(\loc,\mathbf{x})$, if $\mathbf{x}\models\guards(\tau)$. A state $(\loc',\mathbf{x}')$ is a {\em successor} of a state $(\loc,\mathbf{x})$, if there is an enabled transition $\tau=(\loc,\prob)$ at $(\loc,\mathbf{x})$ such that $\prob(\loc') > 0$ and $\mathbf{x}'$ can be obtained by applying the update of $\tau$ to $\mathbf{x}$. A state $(\loc,\mathbf{x})$ is said to be an {\em initial state}, if $\loc = \locinit$ and $\mathbf{x} \in \thetainit$. A state $(\loc,\mathbf{x})$ is said to be a {\em terminal state}, if $\loc = \locterm$.

A {\em finite path} in $\pCFG$ is a sequence $(\loc_0,\mathbf{x}_0),(\loc_1,\mathbf{x}_1),\dots,(\loc_k,\mathbf{x}_k)$ of states with $(\loc_0,\mathbf{x}_0)$ an initial state and each state $(\loc_{i+1},\mathbf{x}_{i+1})$ being a successor of $(\loc_i,\mathbf{x}_i)$. A {\em run} (or an {\em execution}) in $\pCFG$ is an infinite sequence of states whose each finite prefix is a finite path. We use $\Run^\pCFG$ to denote the set of all runs in $\pCFG$. A state $(\loc,\mathbf{x})$ is {\em reachable} in $\pCFG$ if there exists a finite path in $\pCFG$ whose last state is $(\loc,\mathbf{x})$. 

\parag{Semantics.} The pCFG semantics are formalized as (possibly infinite-state) discrete-time Markov chains~\cite{AgrawalC018}. In particular, a pCFG $\pCFG$ together with an initial variable valuation $\vecinit \in \thetainit$ define a discrete-time Markov chain over the set of states in $\pCFG$, whose trajectories correspond to runs in~$\pCFG$. Each trajectory of the Markov chain starts in the initial state $(\locinit,\vecinit)$. Then, at each time step, if the trajectory is in state $(\loc_i,\mathbf{x}_i)$, the next state $(\loc_{i+1},\mathbf{x}_{i+1})$ is defined according to the probability distribution over successor states and the update of the unique pCFG transition $\tau_i$ enabled at $(\loc_i,\mathbf{x}_i)$. This Markov chain gives rise to a probability space $(\Run^\pCFG, \mathcal{F}^\pCFG, \probm^\pCFG_{\vecinit})$ over the set of all runs in $\pCFG$ that start in the initial state $(\locinit,\vecinit)$. The probability space is formally defined via the cylinder construction~\cite{MeynTweedie}. We use $\mathbb{E}^\pCFG_{\vecinit}$ to denote the expectation operator in this probability space.

\parag{Termination.} We say that a run $\rho = (\loc_0,\mathbf{x}_0),(\loc_1,\mathbf{x}_1),\dots$ in the pCFG $\pCFG$ is {\em terminating}, if it reaches some terminal state. We use $\Termset \subseteq \Run^\pCFG$ to denote the set of all terminating runs in $\pCFG$ and use $\timeterm$ as the random variable for the number of steps $\pCFG$ takes before termination. A pCFG $\pCFG$ is said to be {\em almost-surely (a.s.) terminating}, if the probability of a random run terminating is equal to $1$, i.e.~if $\probm_{\vecinit}[\Termset] = 1$ for every initial variable valuation $\vecinit \in \thetainit$. Our algorithm for refuting DP assumes that the underlying pCFG terminates almost-surely. 

\subsection{Differential Privacy}\label{sec:diffprivacy}

We now formally define the notion of differential privacy. Consider a PP with some specified set of inputs and a specified set of output variables. Intuitively, a PP is differentially private if, for any two "similar" input variable valuations $\vecinit^1$ and $\vecinit^2$, the probability of every output event on these two inputs differs by an insignificant amount. For example, suppose that we define the notion of similarity of inputs by the presence and absence of exactly one datapoint in the input. Then, a PP being differentially private ensures that an adversarial observer cannot easily detect information about individual datapoints by observing outputs on similar inputs.

In order to formalize the above intuition, we need to formalize the notions of input similarity and of output distributions. Consider an a.s.~terminating pCFG $\pCFG=(\locs,\pvars,\Vout,\locinit,\thetainit,\transitions,\guards,\updates)$. We say that the set of initial variable valuations $\thetainit \subseteq \R^{|\pvars|}$ is the set of its {\em inputs}, and that the set of all output variable valuations $\R^{\Vout}$ is the set of its {\em outputs}. 

\parag{Similarity of inputs.} To formally reason about the similarity of inputs, we consider a boolean predicate $\Phi(\vecinit^1,\vecinit^2)$ over $\R^\pvars \times \R^\pvars$ which we call a {\em similarity relation}. We say that two inputs $\mathbf{\vecinit^1},\mathbf{\vecinit^2} \in \thetainit$ are {\em similar} if $\Phi(\vecinit^1,\vecinit^2)$ is true, and write $\vecinit^1 \similar \vecinit^2$.

\parag{Output distributions.} For each input $\vecinit \in \thetainit$, an a.s.~terminating pCFG defines a probability distribution over the set of its outputs as follows. Denote by $\mathbf{x}^{\out}$ the projection of a variable valuation $\mathbf{x} \in \R^{|\pvars|}$ to the components of output variables in $\Vout$. For a Borel-measurable set $A \subseteq \R^{|\Vout|}$, define the set of terminating pCFG runs that reach a variable valuation in $A$ upon termination via $\Output(A) = \{ \rho \in \Run^\pCFG \mid \rho \textit{ reaches a terminal state } (\locterm,\mathbf{x}) \text{ with } \mathbf{x}^{\out} \in A\}$.

The {\em output distribution} of $\pCFG$ given an initial variable valuation $\vecinit \in \thetainit$ is a probability distribution over outputs $\R^{|\Vout|}$ that is defined via
$\mu^{\pCFG}_{\vecinit}(A) = \probm^\pCFG_{\vecinit}[\Output(A)],$
for every Borel-measurable set $A \subseteq \R^{|\Vout|}$. We omit the superscript $\pCFG$ whenever it is clear from the context.


\begin{definition}[Differential privacy~\cite{Dwork06,DworkR14}]\label{def:diffprivacy}
    Consider a pCFG $\pCFG=(\locs,\pvars,\Vout,\locinit,\thetainit,\transitions,\guards,\updates)$ which is assumed to be a.s.~terminating, a similarity relation $\Phi$ over its inputs and $\epsilon \geq 0$. We say that the pCFG $\pCFG$ is {\em $\epsilon$-differentially private ($\epsilon$-DP)}, if for every pair of inputs $\vecinit^1,\vecinit^2 \in \thetainit$ such that $\vecinit^1 \similar \vecinit^2$ and for every Borel-measurable set of outputs $A \subseteq \R^{|\Vout|}$ we have
    \begin{equation}\label{eq:diffprivacy}
        \mathbb{P}_{\vecinit^1}\Big[\Output(A)\Big] \leq e^\eps \cdot \mathbb{P}_{\vecinit^2}\Big[\Output(A)\Big],
    \end{equation}
    or, equivalently, $\mu_{\vecinit^1}(A) \leq e^\eps \cdot \mu_{\vecinit^2}(A)$.
\end{definition}

\parag{Problem statement.} We now define the DP refutation problem that we consider in this work. Suppose that $\pCFG=(\locs,\pvars,\Vout,\locinit,\thetainit,\transitions,\guards,\updates)$ is a pCFG, $\Phi$ is a similarity relation over its inputs and $\epsilon\geq 0$. Prove that $\pCFG$ is not $\epsilon$-DP, and compute a pair of inputs $\vecinit^1$ and $\vecinit^2$ such that $\vecinit^1 \similar \vecinit^2$ but for which eq.~\eqref{eq:diffprivacy} is violated.

\begin{remark}
    Note that we use the set-based definition of DP which relates probabilities of events instead of single elements. This is because we are considering PPs with real-valued variables and continuous distributions, which may make the probability of observing a single output element zero. 
\end{remark}
\section{Sound and Complete Proof Rule for Differential Privacy Refutation}\label{sec:proofrule}

We now present our novel proof rule for DP refutation. As outlined in Section~\ref{sec:overview}, our proof rule is based on the notions of upper expectation supermartingales and lower expectation submartingales, which respectively provide proof rules for computing upper and lower bounds on the expected value of a function upon PP termination. To that end, in Section~\ref{sec:ESM} we first recall these notions. We then present our proof rule for DP refutation and establish its soundness and completeness in Section~\ref{sec:proof-rule-ref}. Let $\pCFG=(\locs,\pvars,\Vout,\locinit,\thetainit,\transitions,\guards,\updates)$ be a pCFG. We fix some additional terminology:
\begin{enumerate}
    \item A {\em state function} $\eta$ is a function which to each pCFG location $\loc \in \locs$ assigns a Borel-measurable function $\eta(\loc): \mathbb{R}^{\pvars} \rightarrow \mathbb{R}$ over program variables. We use $\eta(\loc,\mathbf{x})$ for $\eta(\loc)(\mathbf{x})$ as well.

    \item A {\em predicate function} $\Pi$ is a function which to each pCFG location $\loc \in \locs$ assigns a predicate $\Pi(\loc)$ over program variables. The predicate function $\Pi$ naturally induces a set of states $\{(\loc,\mathbf{x}) \mid \mathbf{x} \models \Pi(\loc)\}$ and we also use $\Pi$ to denote this set.
\end{enumerate}

\subsection{Background: Expectation Supermartingales and Submartingales}\label{sec:ESM}

Recall from Section~\ref{sec:overview} that, given a function $f$ over pCFG outputs, upper expectation supermartingales (UESMs) and lower expectation submartingales (LESMs) for $f$ are functions which assign a real value to each pCFG state and which are required to satisfy a set of conditions at reachable pCFG states. These conditions together ensure that U/LESMs give rise to an upper/lower bound on the expected value of $f$ on pCFG output. 

The following definitions formalize the notions of U/LESMs. However, since it is in general not feasible to compute the set of all reachable pCFG states, we impose the defining conditions of U/LESMs over a supporting invariant rather than over reachable pCFG states. Given a pCFG $\pCFG$, an {\em invariant} in $\pCFG$ is a predicate function $I$ which contains all reachable states in the pCFG, i.e.~for every reachable state $(\loc,\mathbf{x})$ we have $\mathbf{x} \models I(\loc)$. This is done with later automation in mind, since invariants can be synthesized by using off-the-shelf tools as we discuss in Section~\ref{sec:algo}.

\begin{definition}[Upper expectation supermartingale~\cite{ChatterjeeGNZ24}]\label{def:uesm}
Let $\pCFG=(\locs,\pvars,\Vout,\locinit,\thetainit,\transitions,\guards,\updates)$ be an a.s.~terminating pCFG, $I$ be an invariant in $\pCFG$, and $f\colon \R^{\pvars_\out} \to \R$ be a Borel-measurable function over the outputs in $\pCFG$. An {\em upper expectation supermartingale (UESM)} for $f$ is a state function $U_f$ such that the following conditions are satisfied: 
\begin{compactitem}
    \item \textit{Zero on output.} For every $\val \models I(\locterm)$, it holds that $U_f(\locterm,\val)=0$.
    \item \textit{Expected $f$-decrease.} For every location $\loc$, transition $\tau = (\loc,\prob) \in\, \mapsto$ and variable valuation $\val \models I(\loc) \wedge \guards(\tau)$, if $\val'$ is the valuation obtained by performing $\tau$, it holds that 
    \[
        U_f(\loc,\val) + f(\val^\out) \geq \sum_{\loc' \in \locs} \prob(\loc') \cdot \E[U_f(\loc', \val') + f(\val'^\out)].
    \]
\end{compactitem}
\end{definition}

\begin{definition}[Lower expectation submartingale~\cite{ChatterjeeGNZ24}]\label{def:lesm}
Let $\pCFG=(\locs,\pvars,\Vout,\locinit,\thetainit,\transitions,\guards,\updates)$ be an a.s.~terminating pCFG, $I$ be an invariant in $\pCFG$, and $f\colon \R^{\pvars_\out} \to \R$ be a Borel-measurable function over the outputs in $\pCFG$. A {\em lower expectation submartingale (LESM)} for $f$ is a state function $L_f$ such that the following conditions are satisfied: 
\begin{compactitem}
    \item \textit{Zero on output.} For every $\val \models I(\locterm)$, it holds that $L_f(\locterm,\val)=0$.
    \item \textit{Expected $f$-increase.} For every location $\loc$, transition $\tau = (\loc,\prob) \in\, \mapsto$ and variable valuation $\val \models I(\loc) \wedge \guards(\tau)$, if $\val'$ is the valuation obtained by performing $\tau$, it holds that 
    \[
        L_f(\loc,\val) + f(\val^\out) \leq \sum_{\loc' \in \locs} \prob(\loc') \cdot \E[L_f(\loc', \val') + f(\val'^\out)].
    \]
\end{compactitem}
\end{definition}

Intuitively, if $U_f$ (resp. $L_f$) is a UESM (resp. LESM) for a function $f$ over the outputs in $\pCFG$, the value of $U_f(\loc,\val)+f(\val^\out)$ (resp. $L_f(\loc,\val)+f(\val^\out)$) is zero on output and for every $\val \models I(\loc)$ it decreases (resp.~increases) in expectation upon every one-step execution of the pCFG.

Unfortunately, the Zero on output and the Expected $f$-decrease/increase conditions alone are not sufficient for U/LESMs to provide sound proof rules for establishing upper or lower bounds on the expected value of a function on program output. However, it was shown in~\cite{ChatterjeeGNZ24} that soundness can be ensured by additionally imposing one of the four {\em OST-soundness} conditions on the program and the U/LESM pair. The name of these conditions is derived from the Optional Stopping Theorem (OST) in martingale theory, which lies at the core of the soundness proof. The first three conditions are derived from the classical OST~\cite{Williams91}, whereas the fourth condition is derived from the extended OST~\cite{Wang0GCQS19}. In what follows, we first provide an intuition behind the four OST-soundness conditions, after which we formally define them in Definition~\ref{def:ost}:
%
\begin{itemize}
    \item[(C1)] \textbf{Bounded Termination Time:} The program terminates within a fixed constant number of steps, effectively capping the execution length.
    \item[(C2)] \textbf{Bounded Values:} The sum of the state functions that define U/LESM and the target function $f$ is bounded by a fixed constant throughout all possible program runs.
    \item[(C3)] \textbf{Bounded Expected Changes:} The program terminates in finite expected time, and the expected step-wise change of the sum of the state functions that define U/LESM and the target function $f$ is bounded by a fixed constant throughout all possible program runs.
    \item[(C4)] \textbf{Probabilistic Termination with Polynomial Step Bounds:} The probability of the program running for a long time decreases exponentially, and the worst-case step-wise change of the sum of the state functions that define U/LESM and the target function $f$ grows at most polynomially in the number of execution steps.
\end{itemize}
\begin{definition}[OST-soundness~\cite{ChatterjeeGNZ24}]\label{def:ost}
    Let $\pCFG$ be a pCFG, $\eta$ be a state function in $\pCFG$ and $f\colon \R^\Vout \to \R$ be Borel-measurable. Moreover, let $\loc_i$ be the $i$-th location and $\val_i$ variable valuation along a random run $\rho$ in $\pCFG$ and $Y_i(\rho) = \eta(\loc_i,\val_i) + f(\val_i^\out)$. The triple $(\pCFG, \eta,f)$ is said to be {\em OST-sound} if  $\E[|Y_i|] < \infty$ holds for every $i \geq 0$, and at least one of the following conditions holds for every $\vecinit \in \thetainit$: 
    \begin{itemize}
        \item[(C1)] $\mathbb{P}_{\vecinit}[\timeterm < c] =1$ for a constant $c$, i.e. $\pCFG$ has uniformly bounded termination time. 
        \item[(C2)] There is a constant $c$ such that, for each time step $i \geq 0$ and every run $\rho$ in $\pCFG$, it holds that $|Y_i(\rho)| < c$, i.e. the sum of $\eta$ and $f$ is uniformly bounded over program runtime.
        \item[(C3)] $\E_{\vecinit}[\timeterm] <\infty$ and there exists a constant $c$ such that for each time step $i \geq 0$ and every run $\rho$ in $\pCFG$, it holds that $\E[|Y_{i+1} - Y_i| \mid \mathcal{F}_i](\rho) < c$, i.e. the expected one-step change of $Y_i$ when conditioned on the past steps is uniformly bounded over program runtime. 
        \item[(C4)] There exist real values $M,c_1,c_2,d$, such that (i) for all $n \in \N$ that are sufficiently large, $\mathbb{P}_{\vecinit}[\timeterm >n] \leq c_1 \cdot e^{-c_2 \cdot n}$, and (ii) for each time step $i \geq 0$ and every run $\rho$ in $\pCFG$ it holds that $|Y_{i+1}(\rho)-Y_i(\rho)| \leq M \cdot n^d$.
    \end{itemize}

\end{definition}

The following theorem shows that U/LESMs for a function $f$ provide sound proof rules for computing upper/lower bounds on the expected value of the function $f$ on pCFG output.

\begin{theorem}[Soundness of U/LESMs \cite{ChatterjeeGNZ24}] \label{thm:esm-soundness}
    Let $\pCFG = (\locs,\pvars,\Vout,\locinit,\thetainit,\transitions,\guards,\updates)$ be an a.s.~terminating pCFG, $I$ be an invariant in $\pCFG$, $f\colon \R^{\pvars_\out} \to \R$ be a Borel-measurable function over the outputs in $\pCFG$, $U_f$ be an UESM for $f$ and $L_f$ be an LESM for $f$. If the triples $(\pCFG,U_f,f)$ and $(\pCFG,L_f,f)$ are OST-sound, then
    \[
    \begin{split}
    U_f(\locinit,\vecinit) + f(\vecinit^\out) &\geq \underset{\val \sim \mu_{\vecinit}}{\E}{[f(\val^\out)]} \\
    L_f(\locinit,\vecinit) + f(\vecinit^\out) &\leq \underset{\val \sim \mu_{\vecinit}}{\E}{[f(\val^\out)]} \\
    \end{split}
    \]
\end{theorem}

The works~\cite{ChatterjeeGNZ24,ChatterjeeGNZ25} show that U/LESMs are not only sound, but also complete proof rule for deriving tight upper and lower bounds on the expected value of a function $f$ on output. Moreover, the computation of U/LEMSs and checking of OST-soundness conditions can be fully automated, which we will use later in Section~\ref{sec:algo} when designing our automated algorithm for $\epsilon$-DP refutation.


\begin{example}[U/LESM for RR-1] \label{ex:esm-RR}
    Consider the randomized response mechanism in Fig.~\ref{fig:RR-1}. The following table shows an example of a UESM $U_f$ and an LESM $L_f$ for $f(\out) = \out^2$:
    \begin{center}
    \begin{tabular}{c|c|c}
        Label & UESM & LESM \\
        \hline 
        $l_1$ & $0.25- \out^2 + 0.5 x^2 $ & $0.25 - \out^2  + 0.5 x^2$ \\
        $l_2$ & $- \out^2 + x^2$ & $-\out^2 + x^2$ \\
        $l_3$ & $0.5 - \out^2$ & $0.5 - \out^2$ \\
        $l_4$ & $-\out^2$ & $-\out^2$ \\
        $l_5$ & $1 - \out^2$ & $1 - \out^2$ \\
        $l_t$ & $0$ & $0$ \\
    \end{tabular}
    \end{center}
    The mechanism trivially satisfies the OST-soundness condition (C1) in Definition~\ref{def:ost}, hence the U/LESM define OST-sound triples. Moreover, the above U/LESM give rise to tight upper and lower bounds on the expected value of $f$ upon termination.
\end{example}

\subsection{Proof Rule for Differential Privacy Refutation} \label{sec:proof-rule-ref}

We are now ready to present our proof rule for $\epsilon$-DP refutation. Consider a pCFG $\pCFG=(\locs,\pvars,\Vout,\locinit,\thetainit,\transitions,\guards,\updates)$, a similarity relation $\Phi$ over its inputs, and a privacy budget $\epsilon \geq 0$. To refute $\epsilon$-DP of the pCFG, i.e.~to prove that $\pCFG$ is {\em not} $\eps$-DP, we need to find a pair of similar inputs $\vecinit^1 \similar \vecinit^2$ and an event over outputs $A \subseteq \R^{\pvars_\out}$ such that 
$\mathbb{P}_{\vecinit^1}[\Output(A)] > e^\epsilon \cdot \mathbb{P}_{\vecinit^2}[\Output(A)].$
As outlined in Section~\ref{sec:overview}, our refutation proof rule explicitly characterizes two similar inputs $\vecinit^1 \similar \vecinit^2$. However, instead of also explicitly characterizing a set of outputs $A$, it characterizes a witness of the existence of such a set $A$ in the form of a non-negative function over the pCFG outputs and a U/LESM pair.

\begin{definition}[Expectation super/sub-martingale witness for $\epsilon$-DP refutation]\label{def:proofrule}
Given a pCFG $\pCFG=(\locs,\pvars,\Vout,\locinit,\thetainit,\transitions,\guards,\updates)$, a similarity relation $\similar$ on $\thetainit$, a privacy budget $\epsilon \geq 0$, and an invariant $I$ in $\pCFG$, an {\em expectation super/sub-martingale witness for $\epsilon$-DP refutation} in $\pCFG$ is a tuple $(\vecinit^1, \vecinit^2, f, U_f, L_f)$, where:
\begin{compactenum}
    \item[(R1)] $\vecinit^1, \vecinit^2 \in \thetainit$ is a pair of similar inputs, i.e.~$\vecinit^1 \similar \vecinit^2$;
    \item[(R2)] $f\colon\R^{\pvars_\out} \to \R$ is a non-negative Borel-measurable function over pCFG outputs; 
    \item[(R3)] $L_f$ is an LESM for $f$ in $\pCFG$ such that the triple $(\pCFG, L_f, f)$ is OST-sound;
    \item[(R4)] $U_f$ is a UESM for $f$ in $\pCFG$ such that the triple $(\pCFG, U_f, f)$ is OST-sound;
    \item[(R5)]  $L_f(\locinit,\vecinit^1)+f(\vecinit^{1 \, \out}) > e^\epsilon \left( U_f(\locinit,\vecinit^2)+f(\vecinit^{2 \, \out}) \right)$.
\end{compactenum}
\end{definition}

Intuitively, our proof rule consists of two similar input points $\vecinit^1 \similar \vecinit^2$, a non-negative function $f$ over outputs, and a witness to show that the expectation of $f$ upon the termination of $\pCFG$ on input $\vecinit^1$ is more than a factor of $e^\epsilon$ larger than its expectation upon the termination of $\pCFG$ on input $\vecinit^2$. This discrepancy between the output distributions on two inputs is witnessed by an LESM $L_f$ and a UESM $U_f$ for $f$ in $\pCFG$ which satisfy the inequality $L_f(\locinit,\vecinit^1)+f(\vecinit^{1 \, \out}) > e^\epsilon \left( U_f(\locinit,\vecinit^2)+f(\vecinit^{2 \, \out}) \right)$. The following theorem establishes that the expectation super/sub-martingale witnesses indeed provide a sound and complete proof rule for refuting $\epsilon$-DP.


\begin{theorem}[Soundness and completeness, Proof in Appendix~\ref{app:ref-sound-complete:proof}] \label{thm:ref-sound-complete}
Given a pCFG $\pCFG=(\locs,\pvars,\Vout,\locinit,\thetainit,\transitions,\guards,\updates)$, a similarity relation $\similar$ on $\thetainit$, a privacy budget $\epsilon \geq 0$, and an invariant $I$ in $\pCFG$, the pCFG $\pCFG$ is {\em not} $\eps$-DP {\em if and only if} there exists an expectation super/sub-martingale witness for $\epsilon$-DP refutation $(\vecinit^1, \vecinit^2, f, U_f, L_f)$.
\end{theorem}

\begin{proof}[Proof Sketch] For soundness, we show that if $\pCFG$ is $\epsilon$-DP, then for every non-negative function $f$ and similar inputs $x_1,x_2$, it holds that $\underset{\val \sim \mu_{\val_1}}{\E}{[f(\val^\out)]} \leq e^\eps \underset{\val \sim \mu_{\val_2}}{\E}{[f(\val^\out)]}$. Hence, if an expectation super/sub-martingale witness exists, then $\pCFG$ is not $\epsilon$-DP, due to Theorem \ref{thm:esm-soundness}.

For completeness, we show that if $\pCFG$ is not $\epsilon$-DP, i.e. a pair of similar inputs $x_1 \similar x_2$ and set of outputs $A$ exist such that $\mu^\pCFG_{\vecinit^1}(A) > e^\eps \mu^\pCFG_{\vecinit^2}(A)$, then the indicator function of $A$ can be extended to an expectation super/sub-martingale witness for $\epsilon$-DP refutation of $\pCFG$. 
\end{proof}

Note that non-negativity of $f$ in condition (R2) is necessary for the soundness of our proof rule as shown in Appendix~\ref{app:ref-sound-complete:proof}.

\begin{example}[Witness for $\epsilon$-DP refutation of RR-1]
    Consider the randomized response mechanism in Fig.~\ref{fig:RR-1}. The tuple $\big((1,0), (0,0), f, U_f, L_f\big)$, with $f, U_f$ and $L_f$ defined as in Example \ref{ex:esm-RR}, is an expectation super/sub-martingale witness for $\epsilon$-DP refutation for every $\epsilon < \ln 3$, since: 
    \begin{compactitem}
        \item[(R1)] $(1,0) \similar (0,0)$ by the definition of the similarity relation in Fig.~\ref{fig:RR-1}.
        \item[(R2)] $f(\out) = \out^2$ is non-negative and Borel-measurable.
        \item[(R3)] $L_f$ is an LESM for $f$ in $\pCFG$ and $(\pCFG,L_f,f)$ is OST-sound as shown in Example \ref{ex:esm-RR}.
        \item[(R4)] $U_f$ is a UESM for $f$ in $\pCFG$ and $(\pCFG,U_f,f)$ is OST-sound as shown in Example \ref{ex:esm-RR}.
        \item[(R5)] $L_f\big(\locinit,(1,0)\big)+f(0)=0.75$ and  $U\big(\locinit,(0,0)\big)+f(0) = 0.25$, hence (R5) holds for every $\epsilon< \ln 3$.
    \end{compactitem}
\end{example}

\section{Automated Algorithm for Differential Privacy Refutation} \label{sec:algo}

We now present our algorithm for automated $\epsilon$-DP refutation, by synthesizing an instance of our expectation super/sub-martingale witness that we introduced in Section~\ref{sec:proofrule}. Since the problem of verification (and hence refutation) of DP properties is known to be undecidable~\cite{BartheCJS020}, we cannot hope for a sound and complete algorithm that is applicable to all PPs. Therefore, we instead aim for a sound and conditionally semi-complete algorithm. By soundness, we mean that our algorithm provides formal guarantees on the correctness of its output. By conditional semi-completeness, we mean that the algorithm is guaranteed to refute DP for a {\em subclass of non-DP PPs} which we will formally characterize.  

\parag{Algorithm assumptions.} Given the undecidable nature of the DP refutation problem, our algorithm imposes several assumptions towards enabling its automation:
\begin{compactenum}
    \item {\em Polynomial programs.} We consider PPs and corresponding pCFGs in which all arithmetic expressions are polynomial expressions over program variables. Furthermore, by introducing dummy variables for expressions appearing in transition guards, we assume that arithmetic expressions appearing in transition guards are linear. The latter assumption is without loss of generality and is imposed for more efficient quantifier elimination (see details below).
    \item {\em Finite and computable moments of sampling instructions.} We assume that all sampling instructions appearing in our PPs have finite and computable moments. That is, for each probability distribution $\delta$ appearing in a sampling instruction and for each $p \in \mathbb{N}$, the $p$-th moment $\mathbb{E}_{X \sim \delta}[|X|^p]$ is finite and can be computed by our algorithm. This assumption holds by default for most standard probability distributions such as Laplace, Exponential, Normal or Uniform distribution, for which look-up tables for moment values are readily available.
    \item {\em Polynomial similarity relation.} We assume that the similarity relation $\Phi(\vecinit^1,\vecinit^2)$ can be expressed by a conjunction of polynomial inequalities over variable valuations $\vecinit^1$ and $\vecinit^2$.
    \item {\em Linear supporting invariants.} Recall, in Definitions~\ref{def:uesm} and~\ref{def:lesm}, we defined U/LESMs with respect to supporting invariants. Our algorithm assumes that such invariants are provided as $I(\loc)$ for every $\loc \in \locs$, and furthermore that $I(\loc)$ is given as a conjunction of linear inequalities over program variables. In practice, and in our prototype implementation, linear supporting invariants can be efficiently and automatically computed by off-the-shelf invariant generation tools for linear and polynomial programs~\cite{aspic,ChatterjeeFGG2020}. When synthesizing invariants for PPs, we treat sampling instructions as non-deterministic assignments.
\end{compactenum}

\parag{Algorithm outline: Template-based synthesis of expectation super/sub-martingale witnesses.} Our algorithm takes as input a pCFG $\pCFG=(\locs,\pvars,\Vout,\locinit,\thetainit,\transitions,\guards,\updates)$ with polynomial arithmetic expressions and sampling instructions that have finite and computable moments, a linear invariant~$I$, a polynomial similarity relation $\Phi$, and a privacy budget $\epsilon \geq 0$. In addition, it takes as input a maximal polynomial degree parameter $D \in \mathbb{N}$ to be used in the synthesis algorithm.

In order to refute $\epsilon$-DP of the pCFG $\pCFG$, the algorithm computes an instance $(\vecinit^1, \vecinit^2, f, U_f, L_f)$ of an expectation super/sub-martingale witness for $\epsilon$-DP refutation. This is done by following the template-based synthesis approach to simultaneously compute each of the five elements of the tuple. The synthesis proceeds in four steps. 

\parag{Step~1: Setting up templates.} The first step of the algorithm consists of setting up symbolic templates for each element of the tuple $(\vecinit^1, \vecinit^2, f, U_f, L_f)$. For two inputs $\vecinit^1$ and $\vecinit^2$, the symbolic templates are defined by introducing two vectors of symbolic variables $\vecinit^1 = (\vecinit^1[1],\dots,\vecinit^1[|\pvars|])$ and $\vecinit^2 = (\vecinit^2[1],\dots,\vecinit^2[|\pvars|])$ which denote the values of each program variable on the two inputs. The template for the function $f$ over pCFG outputs is defined by introducing a symbolic polynomial function over the output variables of degree at most $D$, i.e.
$f(\val^\out) = \sum_{m \in M^D(\Vout)} c_m \cdot m$, 
where $M^D(\pvars)$ is the set of all monomials of degree at most $D$ over $\Vout$ and $c_m$'s are the real-valued symbolic template variables. Finally, the templates for the UESM $U_f$ and the LESM $L_f$ are defined as symbolic polynomial functions over the program variables for each location $\loc \in \locs$ in the pCFG:
\begin{equation}\label{eq:template}
        U_f(\loc,\val) = \begin{cases}
                            \sum\limits_{m \in M^D(\pvars)} s_{\loc,m} \cdot m  & \loc \neq \locterm \\ 
                            0 & \loc = \locterm
                            \end{cases}, \,\,\,\,\,\,\,\,\,\,
        L_f(\loc,\val) = \begin{cases}
                            \sum\limits_{m \in M^D(\pvars)} t_{\loc,m} \cdot m  & \loc \neq \locterm \\ 
                            0 & \loc = \locterm
                            \end{cases},\\ 
    \end{equation}
    where $M^D(\pvars)$ is the set of all monomials of degree at most $D$ defined over $\pvars$ and $s_{\loc,m}$'s and $t_{\loc,m}$'s are the real-valued symbolic template variables. We explicitly require both the UESM and the LESM to be equal to $0$ for $\loc = \locterm$ as this is required by the Zero on output condition of UESMs and LESMs, see Definition~\ref{def:uesm} and Definition~\ref{def:lesm}.
    
    The values of the symbolic template variables at this stage are unknown. The later steps of the algorithm will then compute concrete values for these symbolic variables, which will together give rise to an instance of our $\epsilon$-DP refutation witness. 

\parag{Step~2: Constraint collection.}
    In this step, the algorithm collects a system of constraints over the symbolic template variables which together entail that the tuple $(\vecinit^1, \vecinit^2, f, U_f, L_f)$ defines a valid expectation super/sub-martingale witness for $\epsilon$-DP refutation. To ensure this, we need to collect constraints for the five defining conditions in Definition~\ref{def:proofrule}:
    \begin{compactitem}
        \item[(R1)] The two inputs $\vecinit^1$ and $\vecinit^2$ need to satisfy the similarity relation. This is captured by the constraint $[\vecinit^1 \sim_\Phi \vecinit^2]$. The constraint is collected by substituting the symbolic template variables that define $\vecinit^1$ and $\vecinit^2$ into the conjunction of polynomial inequalities that define the similarity relation $\Phi$ (recall that this was one of the algorithm assumptions).
        \item[(R2)] $f$ must be non-negative on all reachable output evaluations of $\pCFG$. This is captured by the constraint $[\forall \val \in \R^\pvars.\, \val \in I(\locterm) \Rightarrow f(\val^\out) \geq 0]$.
        \item[(R3)] $L_f$ must be an LESM for $f$ and the triple $(\pCFG,L_f,f)$ must be OST-sound. For $L_f$ to be an LESM for $f$, it should satisfy the Zero on output and the Expected $f$-increase conditions in Definition \ref{def:lesm}. The former is satisfied by the definition of $L_f$ in eq.~\eqref{eq:template}. For the latter, for each $\loc \in \locs$ and each transition $\tau=(\loc,\Pr)$, the algorithm collects the following constraint: 
        \begin{equation} \label{eq:lesm}
        \forall \val \in \R^\pvars. \, \val \in I(\loc) \cap \guards(\tau) \Rightarrow L_f(\loc,\val) + f(\val^\out) \leq \sum_{\loc' \in \locs} \prob(\loc') \E[L_f(\loc', \val') + f(\val'^\out)]
        \end{equation}
        where every occurrence of $f$ and $L_f$ is substituted by their symbolic polynomial template. Note that the symbolic expression for the expectation can be explicitly computed by the linearity of the expectation operator and since all moments of the sampling instructions are assumed to be finite and computable.

        The system of constraints that entail that the triple $(\pCFG, L_f, f)$ is OST-sound is analogous to the one constructed in the work~\cite{ChatterjeeGNZ24} on automated computation of UESMs and LESMs. Hence, for the interest of space, we omit the details here and refer the reader to Appendix~\ref{app:algoost}.
        \item[(R4)] $U_f$ must be a UESM for $f$ and the triple $(\pCFG,U_f,f)$ must be OST-sound. Similarly as above, $U_f$ is zero on output by definition in eq.~\eqref{eq:template}. For the expected $f$-decrease condition, the algorithm collects the same constraint as in eq.~\eqref{eq:lesm}, and replaces $\leq$ with $\geq$ and $L_f$ with $U_f$. 
        \item[(R5)] Finally, the algorithm collects the following constraint where $f$, $U_f$ and $L_f$ are substituted by their symbolic polynomial templates: $L_f(\locinit,\vecinit^1) + f(\vecinit^{1 \, \out}) > e^\epsilon \cdot \left( U_f(\locinit,\vecinit^2) + f(\vecinit^{2 \, \out}) \right)$.
    \end{compactitem}

\parag{Step~3: Constraint simplification via $\forall$ quantifier elimination.}
    Looking at the constraints collected in Step~2, we observe that the constraints for conditions (R1) and (R5) are purely existentially quantified polynomial constraints over the symbolic template variables. However, constraints for conditions (R2), (R3) and (R4) are of the form $[\forall \val \in \R^n. \, P(\val) \Rightarrow Q(\val)]$, where $P$ is a conjunction of linear inequalities over $\pvars$ (since we assume that all transition guards in the pCFG and the supporting invariant $I$ are linear) and $Q$ is a polynomial inequality over $\pvars$ (since we assume that all arithmetic expressions in our pCFGs are polynomial). The algorithm then utilizes Handelman's Theorem \cite{handelman} from algebraic geometry to translate each of these constraints into a purely existentially quantified set of linear inequalities over the symbolic template variables as well as auxiliary variables introduced by the translation. The translation is sound in the sense that every solution to the new system of constraints gives rise to a solution of the original system of constraints, while it is also complete under the additional condition that the supporting invariant $I$ defines a closed and bounded set~\cite{AsadiC0GM21}. Since this is a standard procedure in template-based synthesis approaches to program analysis, we omit the details of this translation and refer the reader to~\cite{AsadiC0GM21,ChatterjeeGNZ24}. 

\parag{Step~4: Constraint solving.} 
    Denote by $\Psi$ the resulting system of purely existentially quantified polynomial constraints obtained upon constraint collection in Step~2 and the application of universal quantifier elimination in Step~3. To synthesize an instance $(\vecinit^1, \vecinit^2, f, U_f, L_f)$ of an expectation super/sub-martingale witness for $\epsilon$-DP refutation, the algorithm solves the system of constraints $\Psi$ by employing an off-the-shelf SMT solver. 
    If the SMT solver finds a solution, the algorithm reports ''$\epsilon$-DP refuted'' and returns the witness $(\vecinit^1, \vecinit^2, f, U_f, L_f)$ to the user. Otherwise, the algorithm reports ''Unknown''.


\parag{Algorithm soundness, semi-completeness, complexity.} The following theorem establishes soundness, conditional semi-completeness and an upper bound on the complexity of our algorithm. The soundness and conditional semi-completeness follow by soundness and semi-completeness of all steps of the algorithm. The complexity upper bound follows since the $\epsilon$-DP refutation synthesis problem is reduced to solving a sentence in the existential theory of the reals which can be done in PSPACE. 

\begin{theorem}[Soundness, conditional semi-completeness, complexity, Proof in Appendix~\ref{app:algoproof}] \label{thm:algo-sound-complete}
    Let $\pCFG=(\locs,\pvars,\Vout,\locinit,\thetainit,\transitions,\guards,\updates)$ be a pCFG, $\Phi$ a similarity relation over $\thetainit$, $\epsilon \geq 0$ a privacy budget, and $I$ an invariant in $\pCFG$ that satisfy all the algorithm assumptions. The following claims hold:
    \begin{compactenum}  
        \item {\em (Soundness)} If the algorithm returns a tuple $T=(\vecinit^1,\vecinit^2,f,U_f,L_f)$, then $\pCFG$ is not $\eps$-DP and $T$ is a valid instance of an expectation super/sub-martingale witness for $\epsilon$-DP refutation.
        \item {\em (Conditional Semi-completeness)} If $\pCFG$ is not $\eps$-DP, the invariant $I$ defines a closed and bounded set, and there exists an instance of an expectation super/sub-martingale witness for $\epsilon$-DP refutation $(\vecinit^1,\vecinit^2,f,U_f,L_f)$ where $f$, $L_f$ and $U_f$ can be specified via polynomial expressions, then for a sufficiently large value of the maximum degree parameter $D$, the algorithm is guaranteed to find an instance of an expectation super/sub-martingale witness for $\epsilon$-DP refutation.
        \item {\em (Complexity)} The runtime complexity of the algorithm is PSPACE in the size of the pCFG $\pCFG$, polynomial expressions that specify the similarity relation $\Phi$, the invariant $I$ and the privacy budget $\epsilon$, for any fixed value of the maximum polynomial degree parameter $D$.
    \end{compactenum}
\end{theorem}

\begin{remark}
    We emphasize two points regarding our algorithm:
    \begin{compactenum}
        \item Our algorithm is applicable to both discrete and continuous probabilistic programs. Specifically, in case the input variables of the program are integers, the system of constraints solved in Step~4 will become a mixed-integer system as $\vecinit^1$ and $\vecinit^2$ are required to be integers. 
        \item Given a pCFG $\pCFG$ and technical parameter $D$, our algorithm is fully automated. The invariant $I$ can be generated automatically by state-of-the-art tools such as ASPIC\cite{aspic} (used in our experiments). Additionally, our method could be paired with the template-based invariant generation method of \cite{ChatterjeeFGG2020} to synthesize the invariants and the UESM/LESM certificate simultaneously. 
    \end{compactenum}
\end{remark}




\section{Experimental Evaluation}\label{sec:experiments}

We implemented a prototype of our method which we call SuperDP\footnote{Available at \href{https://github.com/ChatterjeeGroup-ISTA/SuperDP}{https://github.com/ChatterjeeGroup-ISTA/SuperDP} and \href{https://doi.org/10.5281/zenodo.18930113}{https://doi.org/10.5281/zenodo.18930113}}, standing for {\bf Super}martingale-based {\bf D}ifferential {\bf P}rivacy refutation. Our prototype tool takes as input (i) a probabilistic program which is translated into a pCFG $\pCFG$, (ii) a similarity relation $\Phi$, and (iii) a privacy budget $\eps \geq 0$. It then runs the algorithm in Section~\ref{sec:algo} to refute $\eps$-DP of $\pCFG$. If successful, the tool returns an expectation super/sub-martingale witness $(\val_1,\val_2,f,U_f,L_f)$ that refutes $\eps$-DP of $\pCFG$. 

\parag{Research questions.} The goal of our experiments is to answer three research questions (RQs):
\begin{compactitem}
    \item[(RQ1)] {\em Effectiveness.} Is SuperDP effective in refuting $\epsilon$-DP of stochastic mechanisms, including those that are out of the reach of existing methods?
    \item[(RQ2)] {\em Efficiency.} Is SuperDP able to refute $\epsilon$-DP at runtimes competitive with the state of the art?
    \item[(RQ3)] {\em Limitations.} What are the practical limitations of our approach?
\end{compactitem}

\subsection{Experimental Setup}

\parag{Implementation.} SuperDP is implemented in Java and uses Z3 \cite{z3} and MathSAT5 \cite{mathsat} as backend SMT-solvers to solve the system of polynomial constraints obtained in Step~4 of our algorithm in Section~\ref{sec:algo}. It also uses ASPIC~\cite{aspic} to generate supporting linear invariants. Given an input as stated above, SuperDP runs the algorithm in Section~\ref{sec:algo} with maximal polynomial degrees $D=1,\dots, 6$. We note that our algorithm requires the number of program variables to be finite and statically bounded. Hence, if the input PP contains an array whose size is an input parameter (as do the first 7 benchmarks in Table \ref{tab:exp-results}), we start the DP refutation analysis with array size $2$. We then iteratively increment the array size by $1$ and run our algorithm in Section~\ref{sec:algo} until we reach an array size for which SuperDP successfully refutes $\epsilon$-DP, or until timeout is reached ($5$ minutes per benchmark). 

\parag{Baselines.} We compare the performance of our prototype against two state of the art tools for reasoning about DP that support $\epsilon$-DP refutation:
\begin{compactitem}
    \item {\em Static analysis baseline.} CheckDP \cite{WangDKZ20} is the only available static analysis tool that is fully automated and supports $\epsilon$-DP refutation, hence we use it as the first baseline. CheckDP tries to find an alignment between random variables along executions on a pair of similar inputs. If an alignment is found, the program is guaranteed to satisfy $\epsilon$-DP; otherwise, CheckDP refutes $\epsilon$-DP by computing a triple $(\vecinit^1,\vecinit^2,A)$ with $\mathbb{P}_{\vecinit^1}[\Output(A)] > e^\eps \cdot \mathbb{P}_{\vecinit^2}[\Output(A)]$. 
    \item {\em Dynamic analysis baseline.} StatDP \cite{DingWWZK18} is a statistical framework for DP refutation. It carefully searches for a pair of similar inputs that might violate the DP property, runs the program many times on those inputs, and evaluates them by hypothesis testing and p-value computation. Smaller p-values indicate stronger evidence against the DP claim. In what follows, we assume that StatDP refutes $\epsilon$-DP if the p-value it returns is less than $0.1$.
\end{compactitem}

\parag{Other tools.} We also attempted to run DiPC \cite{BartheCJS020} as a baseline, but it always terminated with runtime errors, so we are unable to report results for it. We note, however, that only three of our benchmarks (\texttt{RR-1},\texttt{RR-2}, and \texttt{lowprob}) are compatible with their supported syntax. On the other hand, DP-Sniper \cite{BichselSBV21} and DP-Finder \cite{BichselGDTV18} are statistical methods that require huge computational resources (they use 500 GB of memory and 128 CPU cores in their experiments), hence we omit them from our experiments.

\subsection{Benchmarks and Setup}

All experiments were run on an Ubuntu 24.04 machine with an 11th Gen Intel Core i5 CPU and 16 GB RAM with a timeout of 5 minutes. Our prototype and CheckDP use only one core of the CPU, whereas StatDP was allowed to use up to 8 cores. For benchmarks, we consider 15 probabilistic programs collected from the DP literature and run our prototype and the two baselines on them:
\begin{compactitem}
    \item The first 9 benchmarks are standard stochastic mechanisms from the differential privacy literature \cite{DworkR14} (first 9 rows in Table \ref{tab:exp-results}, where \texttt{SVT} stands for "Sparse Vector Technique" and \texttt{RR} stands for "Randomized Response"). \texttt{BadSmartSum} is an incorrect variant of \texttt{SmartSum} introduced in \cite{WangDKZ20}. The parameter $M$ in both \texttt{SmartSum} and \texttt{BadSmartSum} benchmarks is set to $3$. 
    We consider two versions of the randomized response mechanism, \texttt{RR-1} is shown in Fig. \ref{fig:RR-1} and \texttt{RR-2} is shown in Fig. \ref{fig:RR-2} of Appendix~\ref{app:benchmarks}.
    \item The \texttt{geometric} mechanism of Fig. \ref{fig:geometric} (in Appendix~\ref{app:benchmarks}) first samples the noise variable $z$ from a geometric distribution (lines $l_2$ to $l_6$) and $\mathit{sign}$ from a Bernoulli and adds $(-1)^{\mathit{sign}}\cdot z$ to the input. This mechanism is particularly interesting because the geometric sampling is done by an almost-sure terminating loop that is not statically bounded. 
    \item The \texttt{PrivBernoulli} mechanism was considered in LightDP \cite{ZhangK17} and is shown in Fig.~\ref{fig:privbernoulli} (in Appendix~\ref{app:benchmarks}). 
    Intuitively, given and input $x$, it returns the result of an $x$-biased coin toss. We consider two variants of this mechanism: (i) a variant where the input $x$ can take any value in the range $[0,1]$, which provides no DP guarantees, and (ii) a variant where $x$ is limited to $[\frac{1}{3},\frac{2}{3}]$, which satisfies $(\ln 2)$-DP. CheckDP does not support these benchmarks, since all randomness in their input has to come from Laplacian samplings. 
    \item The \texttt{lowprob} benchmark is motivated by the StatDP paper \cite{DingWWZK18} and is shown in Fig.~\ref{fig:lowprob} (in Appendix~\ref{app:benchmarks}). Given any input, this mechanism returns 0 with probability at least $1-10^{-6}$. This means that a statistical sampling-based method requires many samples with $x=1$ as input in order to observe $1$ on output. Without such an observation, it cannot differentiate between \texttt{lowprob} and a mechanism that always returns $0$ and satisfies $0$-DP. 
    \item \texttt{RE} stands for "Random element" which is a toy example that, given an array of length 5, returns one of the elements of the array uniformly at random, hence trivially violating privacy. It satisfies $(\epsilon,\delta)$-DP for $\epsilon = 0$ and $\delta = 0.2$, but not $\eps$-DP for any value of $\eps$.
    %
    \item The \texttt{Uniform noise} mechanism \cite{GengV16} is similar to the histogram mechanism, with the only difference being that it adds $\mathit{Uniform}_{[-0.1,0.1]}$ noise rather than Laplacian noise to the elements of the input array.
\end{compactitem}

Among our benchmarks, the \texttt{geometric} benchmark satisfies condition (C4) of Definition \ref{def:ost}, while all others satisfy condition (C1) of OST-soundness (statically bounded termination). 


\begin{table}[t] 
\centering
\resizebox{!}{3cm}{
\texttt{
\begin{tabular}[!t]{|c|c|ccc|cc|cc|}
    \hline
    \multirow{2}{*}{Benchmark} & \multirow{2}{*}{DP} & \multicolumn{3}{c|}{SuperDP} & \multicolumn{2}{c|}{CheckDP \cite{WangDKZ20}} & \multicolumn{2}{c|}{StatDP \cite{DingWWZK18}} \\
    && Refuted DP & Time & $|M^D|$ & Refuted DP & Time & Refuted DP & Time \\
    \hline
    \hline 
    PartialSum & 1 & \cellcolor{green!30} 0.9 & 1.3 & 70 & \cellcolor{green!30} 0.9 &  \cellcolor{green!30} 0.7 &  \cellcolor{green!30} 0.9 & 3.4 \\
    \hline 
    Histogram & 1 & \cellcolor{green!30} 0.9 & \cellcolor{green!30} 0.3 & 35 & \cellcolor{green!30} 0.9 & 0.4 & TO & TO\\
    \hline
    SmartSum & 2 & \cellcolor{green!30} 1.9 &  131.1 & 210 & \cellcolor{green!30} 1.9 &  \cellcolor{green!30} 1.5 & TO & TO\\
    BadSmartSum & $\infty$ & \cellcolor{green!30} >15 &  \cellcolor{green!30}0.4 & 28 & \cellcolor{green!30} >15 & 0.9 & TO & TO \\
    \hline 
    NoisyMax & 1 & TO & TO & - & \cellcolor{green!30} 0.9 & 13.7 & \cellcolor{green!30} 0.9 & \cellcolor{green!30} 12.4 \\
    \hline 
    SVT & 1 & TO & TO & - & \cellcolor{green!30} 0.9 & 60.1 & \cellcolor{green!30} 0.9 & \cellcolor{green!30} 52 \\ 
    \hline 
    Gaussian & $\infty$ & \cellcolor{green!30} >15 & \cellcolor{green!30} 0.9 & 15 & NS & NS & TO & TO \\
    \hline 
    RR-1 & $\ln 3$ & \cellcolor{green!30} 1& \cellcolor{green!30} 0.7 & 3 & \cellcolor{gray!40} >15 & \cellcolor{gray!40} 0.9 & \cellcolor{green!30} 1 & 1.9 \\
    RR-2 & $\ln 1.5$ & \cellcolor{green!30} 0.4 & \cellcolor{green!30} 0.4 & 3 & \cellcolor{gray!40} >15 & \cellcolor{gray!40} 0.8 & \cellcolor{green!30} 0.4 & 1.8 \\
    \hline 
    geometric & $\ln 2$ & \cellcolor{green!30} 0.6 & \cellcolor{green!30} 0.7 & 15 & F & F & \cellcolor{gray!40} 1.3 & \cellcolor{gray!40} 17.8\\
    \hline 
    PrivBernoulli-1 & $\infty$ &  \cellcolor{green!30} >15 & \cellcolor{green!30} 0.1 & 3 & NS & NS & 11.5 & 1.9 \\
    PrivBernoulli-2 & $\ln 2$ &  \cellcolor{green!30} 0.6 & \cellcolor{green!30} 0.1 & 3 & NS & NS & 0 & 12.1 \\
    \hline 
    lowprob & $\infty$ & \cellcolor{green!30} >15 & \cellcolor{green!30} 0.1 & 3 & F & F & F & F \\
    \hline 
    RE & $\infty$ & \cellcolor{green!30} >15 & \cellcolor{green!30} 0.1 & 28 & \cellcolor{green!30} >15 & 1.2 & F & F \\
    \hline 
    Uniform noise & $\infty$ & \cellcolor{green!30} >15 & \cellcolor{green!30} 2.2 & 84 & NS & NS & F & F \\ 
    \hline 
\end{tabular}
}
}
\caption{Summary of our experimental results. The "\texttt{DP}" column shows the theoretically proved smallest value of $\epsilon$ for which a mechanism is $\epsilon$-DP, where $\infty$ means that the benchmark is not $\eps$-DP for any $\eps$. The "\texttt{Refuted DP}" columns show the largest $\epsilon$ for which $\epsilon$-DP was successfully refuted by each tool (larger is better), the "\texttt{Time}" columns show the runtime of each tool in seconds, and the $|M^D|$ column indicates the size of the template used in SuperDP. Green cells indicate the tool with the best performance on each benchmark. "\texttt{NS}" stands for "Not Supported"\textsuperscript{1}, "\texttt{TO}" for "Timeout", and "\texttt{F}" for "Fail". Grey cells indicate some incorrect answers by the baseline tools that we observed in our experiments.\textsuperscript{2}}
\label{tab:exp-results}
\end{table}

\newcommand{\customfootnotetext}[2]{{
  \renewcommand{\thefootnote}{#1}
  \footnotetext[0]{#2}}}

\customfootnotetext{1}{The syntax of CheckDP only supports sampling instructions from Laplacian distributions, but not from e.g.~Gaussian, Uniform, or Bernoulli distributions. Hence, some of our benchmarks are not supported by CheckDP. To make CheckDP applicable to benchmarks with probabilistic branching, when providing inputs to it, we replace every instruction of the form "\texttt{if prob($c$)}" with "\texttt{$tmp$ := Lap(1); if $tmp \geq t$}" for an appropriate value of $t$ such that $\Pr[Lap(1) \geq t] \simeq c$.}

\customfootnotetext{2}{On the randomized response benchmarks (\texttt{RR-1} and \texttt{RR-2}), SuperDP and StatDP refute $\epsilon$-DP by computing correct and tight bounds on $\epsilon$. However, the result of CheckDP on these two benchmarks is unsound as the reported value of $\epsilon$ for which $\epsilon$-DP is refuted exceeds the theoretically known value of $\epsilon$ for which the mechanism is $\epsilon$-DP. We report this as a subtle bug in their implementation of counterexample validation, which we believe to be due to these examples not admitting any random variable alignment hence CheckDP refutes $\epsilon$-DP for any value of $\epsilon$. 

}

\subsection{Results}
Table \ref{tab:exp-results} summarizes our experimental results. Green cells indicate the best performing tools on each benchmark, e.g. on the \texttt{PartialSum} benchmark, all the tools could refute up to $0.9$-DP and CheckDP was the fastest taking 0.7s. We discuss our results by considering each of the research questions (RQ1-3) separately.

\begin{compactitem}
    \item[(RQ1)] The first question is concerned with the effectiveness of our method for refuting $\epsilon$-DP. We argue that the answer to this question is positive by analyzing the \texttt{Refuted DP} columns in Table \ref{tab:exp-results}. It can be seen that SuperDP refutes $\epsilon$-DP with the largest value of $\epsilon$ compared to the baselines on 13/15 benchmarks (the 2 failed instances are discussed under RQ3 below). Specifically, on the 9 benchmarks that satisfy $\epsilon$-DP for some finite value of $\epsilon$, SuperDP is able to refute $\epsilon$-DP with a $0.1$-tight value of $\epsilon$ for 7 instances, while CheckDP and StatDP refute such tight bounds for only 5 instances. On the remaining 6 benchmarks that do not satisfy $\epsilon$-DP for any finite value of $\epsilon$, SuperDP successfully refutes up to $15$-DP for all instances, while CheckDP can prove similar bounds for only 2 and StatDP for none.
    
    We note that the result of StatDP on the \texttt{geometric} benchmark is unsound. We believe this is due to the statistical nature of the method, hence the tool can output false-positive and false-negative results with small but non-negligible probability. Moreover, on the \texttt{lowprob} benchmark, StatDP cannot refute $\epsilon$-DP for any $\epsilon$ since it requires many samples to observe a single non-zero output. On the other hand, our SuperDP does not suffer from this limitation and is able to refute $\epsilon$-DP tightly. We also note that CheckDP exited with a runtime error in two instances (\texttt{geometric} and \texttt{lowprob}) without proving or disproving DP. 
    \item[(RQ2)] The second question asks about the efficiency of our approach. By analyzing the \texttt{Time} columns in Table \ref{tab:exp-results}, we argue that the answer to this question is also positive, i.e., SuperDP's runtime is competitive, and in many cases better, than the baselines. Specifically, SuperDP outperforms the baselines in terms of runtime in 11/15 cases. On the \texttt{SmartSum} benchmark, our approach takes more than 2 minutes to terminate because the benchmark has more pCFG locations (larger $|\locs|$) compared to \texttt{PartialSum} and \texttt{Histogram}, which results in larger template sizes and a larger system of constraints in the last step of our algorithm. 
    \item[(RQ3)] The third question is concerned with understanding practical limitations of our approach. Based on the results presented in Table \ref{tab:exp-results}, we observe the following limitations: 
    \begin{compactitem}
        \item {\em Symbolic reasoning about probabilities of \texttt{if}-branching.} Our tool timeouts when solving \texttt{NoisyMax} and \texttt{SVT} benchmarks, while both baselines can refute up to $0.9$-DP. This is because both benchmarks contain code fragments that look like
    \begin{align*} \label{bad-fragment}
        \begin{split}
        &x := Lap(1); \\
        &\texttt{if}(x\geq y) ~~~\dots
    \end{split}\end{align*}
    However, reasoning about such \texttt{if}-statements by using polynomial templates is difficult, since the probability of $(x \geq y)$ needs to be symbolically computed. Hence, our method sometimes struggles with examples that require explicit reasoning about probabilities of conditional branching that cannot be captured by polynomial templates.
    \item {\em Large systems of constraints.} As mentioned above, in cases like \texttt{SmartSum}, our approach needs to solve large systems of constraints which may be computationally expensive for SMT-solvers. Moreover, in the \texttt{Uniform noise} benchmark, our approach requires degree 6 polynomial templates to refute $15$-DP. Using high-degree polynomial templates results in larger systems of constraints, which may be harder for SMT-solvers to solve. 
    \end{compactitem}
\end{compactitem}

\parag{Summary.} As demonstrated by our results in Table \ref{tab:exp-results}, SuperDP can successfully refute $\epsilon$-DP for 13/15 benchmarks, at runtimes that are consistently lower than 3 seconds in all but one case. This clearly positively answers our first two research questions (RQ1) and (RQ2). We also discussed two practical limitations of our approach in answer to (RQ3). The first limitation arises in examples that require symbolic reasoning about probabilities of \texttt{if}-branching which may be a challenge when using polynomial templates, while the second limitation is due to instances where large systems of constraints need to be solved.

\section{Related Work}\label{sec:relatedwork}

Automated analysis of DP has seen a lot of interest in the verification literature. Below we discuss the literature related to refutation of DP. In Appendix \ref{app:related-works} we discuss the works on verification of DP properties.  We classify existing methods for DP refutation into (i) static analysis methods (including our work), and (ii) dynamic (sampling-based) methods.

\parag{Static analysis methods.}
These methods try to refute DP by statically analyzing the program rather than executing it. DiPC \cite{BartheCJS020} considers programs with finite inputs and outputs and does not allow assignment to real or integer variables inside a while loop. They model the given program $M$ as a finite-state discrete-time Markov chain (DTMC) and exhaustively search for similar inputs $a,a'$, and an output $b$ such that $\Pr[M(a) = b] > e^\eps \Pr[M(a')=b]$. Due to the restriction to finite domain inputs and outputs, they have to consider discretizations of programs when they contain infinite-ranges of inputs or outputs. In \cite{FarinaCG21}, the input program, which is limited to have countable inputs and outputs, is executed in a relational and symbolic way. After the execution, the method checks if there exists a coupling between the generated outputs. They also introduce several strategies to look for counterexamples for disproving DP in case their method fails to find a coupling. In contrast, our method does not impose any restrictions on the domains of program variables or places where assignments can occur in the program.

In a separate line of work, CheckDP \cite{WangDKZ20} is the successor of LightDP \cite{ZhangK17} and ShadowDP \cite{WangDWKZ19}, which is fully automated with a verify-invalidate loop that looks for randomness alignments for proving DP. If in any loop iteration, the proposed randomness alignment is invalidated by a counterexample that cannot be fixed by a new alignment, the counterexample is passed to PSI~\cite{GehrMV16} to make sure it is a valid counterexample for the claimed DP bound. The main limitations of CheckDP compared to our approach are: (i) the only source of randomness in their syntax is sampling from the Laplace distribution, (ii) they provide no completeness guarantees, and (iii) we encountered a subtle bug in their implementation of counterexample validation using PSI. In contrast, SuperDP provides soundness and conditional semi-completeness guarantees as explained in Section~\ref{sec:proofrule} and Section~\ref{sec:algo}, and supports general probability distributions including Laplace, Gaussian, and Bernoulli. 


\parag{Dynamic analysis methods.} 
In contrast to static analysis approaches that refute DP without running the input program, statistical methods execute the program multiple times and try to infer/refute DP based on the samples that they observe. StatDP~\cite{DingWWZK18} is one such method that considers input pairs satisfying specific patterns and runs the program many times on them. It then conducts hypothesis testing on the observed samples and computes a p-value, which is a probabilistic estimate of how likely it is that the program satisfies the claimed DP property. A low p-value indicates that the input program most likely does not satisfy the claimed DP. DP-Finder~\cite{BichselGDTV18} uses samples taken from the program to approximate the function $\epsilon(x,x',A)$ as the amount of privacy loss corresponding to similar inputs $x,x'$ and output set $A$. They then find a pair $(x,x')$ that maximizes the approximated privacy loss function and lastly use PSI to compute the exact value of $\epsilon(x,x',A)$. DP-Sniper~\cite{BichselSBV21} uses classifiers (e.g. logistic regression or neural networks) to compute approximations of $\epsilon(x,x',A)$. Although these methods only require a black-box access to the input program, which is an advantage, they suffer from the following limitations: (i) they are all based on approximations and do not provide formal soundness guarantees (e.g. StatDP fails on \texttt{lowprob} and \texttt{geometric} benchmarks in our experiments), (ii) they might require many samples in order to make their computations more accurate which requires huge computing power (e.g. DP-Finder's experiments were conducted with 500 GB RAM and 128 CPU cores). Our approach, in contrast, requires white-box access to the program and provides formal soundness and semi-completeness guarantees, while requiring modest computational resources.


\parag{Martingale-based probabilistic program analysis.} Our method leverages upper expectation supermartingales and lower expectation submartingales of~\cite{ChatterjeeGNZ24,ChatterjeeGNZ25} to formally and automatically reason about DP. Supermartingale and submartingale processes from probability theory~\cite{Williams91} have been extensively used to design fully automated approaches for probabilistic program and model analysis, for properties such as termination and reachability~\cite{ChakarovS13,ChatterjeeFNH18,AgrawalC018,McIverMKK18,ChenH20,AbateGR20,ChatterjeeGMZ22,ChatterjeeGNZZ23,TakisakaOUH21,MajumdarS24,MajumdarS25}, safety~\cite{PrajnaJP07,ChatterjeeNZ17,MathiesenCL23,WangYFLO24,ZhiWLOZ24}, reach-avoidance~\cite{0001LZF21,ZikelicLHC23,ZikelicLVCH23,XUE2025105368,XUE2026112919}, cost analysis~\cite{NgoC018,Wang0GCQS19,ChatterjeeGMZ24}, sensitivity~\cite{WangFCDX20}, and more recently general omega-regular properties~\cite{AbateGR24,HenzingerMSZ25,AbateGR25}. Similarly to our method, many works automate computation of supermartingale and submartingale witnesses by following a template-based synthesis approach. In particular, our algorithm is also related to the work~\cite{ZikelicCBR22} on differential cost analysis in non-probabilistic programs, where the difference in cost usage is analyzed by simultaneously computing polynomial bounds on cost usage in two programs. To the best of our knowledge, our work is the first to use supermartingale and submartingale witnesses, as well as simultaneous reasoning about two expectation bounds of a function, towards analyzing DP.


\section{Conclusion}\label{sec:conclusion}
 
We discuss the key contributions of our method along with its limitations. We present a novel method for $\epsilon$-DP refutation in probabilistic programs by reasoning about expectation mismatch of a non-negative function over program outputs. We introduce a sound and complete proof rule for $\epsilon$-DP refutation based on expectation supermartingale and submartingale
witnesses. Our approach reduces the $\epsilon$-DP refutation problem to finding a non-negative function over outputs with significant expected value difference on two similar inputs, circumventing the need to directly reason about the probability of output events. This enables us to design a fully automated, sound and semi-complete algorithm for synthesizing witnesses for $\epsilon$-DP refutation.

In summary, the key theoretical novelty is to ensure all four desirable properties discussed in Section~\ref{sec:intro}. Moreover, our experiments demonstrate that SuperDP can effectively refute $\epsilon$-DP for a wide range of challenging examples, including those that were beyond the reach of previous methods. This shows the practical potential of our approach. 
The limitations of our work can be summarized as (i) our approach for $\epsilon$-DP refutation is limited to polynomial arithmetic programs and its extension to
other classes of programs remains open; and (ii) our approach is limited to $\epsilon$-DP refutation and its extension to other notions of privacy such as $(\epsilon, \delta)$-DP refutation and to DP verification is another open question. Addressing these are interesting directions of future work.



\begin{acks}
The authors would like to thank Petr Novotn\'y for valuable discussions that helped shape this work. This research was supported by the Singapore Ministry of Education (MOE) Academic Research Fund (AcRF) Tier 1 grant (Proposal ID: 25-SIS-SMU-009), Vienna Science and Technology Fund (WWTF), State of Lower Austria [Grant ID 10.47379/ICT25017], ERC CoG 863818 (ForM-SMArt), and Austrian Science Fund (FWF) 10.55776/COE12.
\end{acks}

\bibliographystyle{ACM-Reference-Format}
\bibliography{bibliography}

\pagebreak
\appendix
\section{Proof of Theorem \ref{thm:ref-sound-complete}} \label{app:ref-sound-complete:proof}
\begin{theorem*}[Soundness and completeness]
Given a pCFG $\pCFG=(\locs,\pvars,\Vout,\locinit,\thetainit,\transitions,\guards,\updates)$, a similarity relation $\similar$ on $\thetainit$, a privacy budget $\epsilon \geq 0$, and an invariant $I$ in $\pCFG$, the pCFG $\pCFG$ is {\em not} $\eps$-DP {\em if and only if} there exists an expectation super/sub-martingale witness for $\epsilon$-DP refutation $(\vecinit^1, \vecinit^2, f, U_f, L_f)$.
\end{theorem*}

\begin{proof}[Proof.] We prove the theorem claim by proving soundness and completeness separately.
\begin{compactenum}
\item {\em Soundness.} Suppose that the pCFG admits an expectation super/sub-martingale witness for $\epsilon$-DP refutation $(\vecinit^1, \vecinit^2, f, U_f, L_f)$. To establish soundness, we need to prove that the pCFG $\pCFG$ is not $\epsilon$-DP. We prove this claim by contradiction. 

Suppose that, on the contrary, $\pCFG$ is $\epsilon$-DP. Then, by Definition~\ref{def:diffprivacy}, since $\vecinit^1 \similar \vecinit^2$ are similar inputs, it must hold that for every event $A \subseteq \R^{\pvars_\out}$ we have $\mu_{\val_1}[A] \leq e^\epsilon \mu_{\val_2}[A]$. Therefore, if we compare the expected value of the function $f$ on the pCFG output given these two inputs (recall, $f$ is an element of our $\epsilon$-DP refutation $(\vecinit^1, \vecinit^2, f, U_f, L_f)$), by the definition of Lebesgue integration we have
\begin{equation}\label{eq:eq1}
    \begin{split}
         \underset{\val \sim \mu_{\val_1}}{\E}{[f(\val^\out)]} =&  \int_{\R^{\pvars_\out}} f(\val) d\mu_{\val_1}(\val) \\
         =& \sup \{\sum_{i=1}^k a_i \cdot \mu_{\vecinit^1}(B_i) \mid \{B_1, \dots, B_k\} \text{ is a partition of } \R^{\pvars_\out} \wedge a_i = \inf_{\val \in B_i} f(\val)\} \\ 
         \leq & \sup \{\sum_{i=1}^k a_i \cdot e^\eps \mu_{\vecinit^2}(B_i) \mid \{B_1, \dots, B_k\} \text{ is a partition of } \R^{\pvars_\out} \wedge a_i = \inf_{\val \in B_i} f(\val)\} \\ 
         =& e^\eps \underset{\val \sim \mu_{\val_2}}{\E}{[f(\val^\out)]}.
    \end{split}
\end{equation}
The inequality in the third line above holds for each partition $\{B_1,\dots,B_k\}$ of $\R^{\pvars_\out}$ since each $a_i \geq 0$ due to the non-negativity of $f$ and since we established that $\mu_{\val_1}[A] \leq e^\epsilon \mu_{\val_2}[A]$ holds for every output event $A \subseteq \R^{\pvars_\out}$, hence the inequality also holds between the two suprema over all possible partitions.

On the other hand, since $(\vecinit^1, \vecinit^2, f, U_f, L_f)$ is assumed to be an expectation super/sub-martingale witness for $\epsilon$-DP refutation, by condition (R3) and (R4) in Definition~\ref{def:proofrule} we know that $U_f$ is a UESM for $f$, $L_f$ is an LESM for $f$, and $(\pCFG,U_f,f)$ and $(\pCFG,L_f,f)$ are OST-sound triples, hence by Theorem~\ref{thm:esm-soundness} it holds that 
\begin{equation*}
\begin{split}
    L_f(\locinit,\vecinit^1) + f(\vecinit^{1 \, \out}) \leq &\underset{\val \sim \mu_{\vecinit^1}}{\E}{[f(\val^\out)]}, \\
    &\underset{\val \sim \mu_{\vecinit^2}}{\E}{[f(\val^\out)]} \leq U_f(\locinit,\vecinit^2) + f(\vecinit^{2 \, \out}).
\end{split}
\end{equation*}
Combining these with eq.~\eqref{eq:eq1}, we conclude that
\[ L_f(\locinit,\vecinit^1) + f(\vecinit^{1 \, \out}) \leq e^\eps \left(U_f(\locinit,\vecinit^2) + f(\vecinit^{2 \, \out})\right). \]
But this contradicts condition (R5) in Definition~\ref{def:proofrule}. Hence, we reach contradiction and the soundness claim in the theorem follows.

\item {\em Completeness.} Suppose now that the pCFG $\pCFG$ is not $\epsilon$-DP. To establish completeness, we need to prove that there exists an expectation super/sub-martingale witness for $\epsilon$-DP refutation.

Since $\pCFG$ is not $\epsilon$-DP, by Definition~\ref{def:diffprivacy} there exist similar inputs $\vecinit^1 \similar \vecinit^2$ and an output event $A \subseteq \R^{\pvars_{\out}}$ such that $\mu^\pCFG_{\vecinit^1}(A) > e^\eps \mu^\pCFG_{\vecinit^2}(A)$. Let $f: \R^{\pvars_{\out}} \rightarrow \mathbb{R}$ be the indicator function of $A$, i.e. 
\[ f(\val^\out) = \begin{cases}
    1, &\text{if } \val^\out \in A \\
    0, &\text{if } \val^\out \notin A.
\end{cases} 
\]
Furthermore, let $L_f$ and $U_f$ be state functions in the pCFG $\pCFG$ defined via
    \[
    \begin{split}
        U_f(\loc,\val) = L_f(\loc,\val) = \begin{cases}
            \probm^{\pCFG(\loc,\val)}[\Output(A)] - f(\val^\out), &\text{if } \loc \neq \locterm \\ 
            0, &\text{if } \loc = \locterm
        \end{cases}
    \end{split}
    \]
    where we use $\pCFG(\loc,\val)$ to denote the pCFG identical to the $\pCFG=(\locs,\pvars,\Vout,\locinit,\thetainit,\transitions,\guards,\updates)$ with the initial location $\locinit = \loc$ and the initial variable valuation set $\thetainit = \{\mathbf{x}\}$, and $\probm^{\pCFG(\loc,\val)}$ denotes the probability measure associated to this pCFG.

    We claim that the tuple $(\vecinit^1, \vecinit^2, f, U_f, L_f)$ is an expectation super/sub-martingale witness for $\epsilon$-DP refutation. To prove this, we need to show that it satisfies all the conditions (R1)-(R5) in Definition~\ref{def:proofrule}: 
    \begin{compactenum}
    \item[(R1)] $\vecinit^1 \similar \vecinit^2$ are similar inputs by our choice of $\vecinit^1,\vecinit^2,A$, hence (R1) is satisfied.
    \item[(R2)] $f$ is an indicator function of the output event $A$, hence it is Borel-measurable and non-negative and (R2) is satisfied.
    \item[(R3) and (R4)] We show that $L_f$ is an LESM for $f$ and that the triple $(\pCFG,L_f,f)$ is OST-sound. The proof for $U_f$ is analogous. 
    To see that $L_f$ is an LESM for $f$, we need to show that it satisfies Zero on output and Expected $f$-decrease conditions in Definition~\ref{def:lesm}. The Zero on output condition is satisfied since, for each output state $(\locterm,\mathbf{x})$, we have $\probm^{\pCFG(\locterm,\val)}[\Output(A)] = f(\val^\out) = 1$ if $\mathbf{x} \in A$ and $\probm^{\pCFG(\locterm,\val)}[\Output(A)] = f(\val^\out) = 0$ otherwise. Hence, $L_f(\locterm,\val) = \probm^{\pCFG(\locterm,\val)}[\Output(A)] - f(\val^\out) = 0$ for all terminal states. On the other hand, to see that the Expected $f$-decrease condition is satisfied, note that by our definition of $L_f$ above, for each $\loc \in \locs \setminus \{\locterm\}$  and $\val \in \R^{\pvars}$ it holds that $L_f(\loc,\val)+f(\val^\out)$ is equal to the probability of reaching an output in $A$ in $\pCFG$ starting from the state $(\loc,\val)$. Therefore, for every enabled transition $\tau = (\loc,\prob)$, we have
    \[
        L_f(\loc,\val) + f(\val^\out) = \sum_{\loc' \in \locs} \prob(\loc') \E[L_f(\loc', \val') + f(\val'^\out)]
    \]
    due to the standard flow conservation of reachability probabilities, and the Expected $f$-decrease condition is satisfied. This shows that $L_f$ is an LESM for $f$ in $\pCFG$. Finally, the OST-soundness of the triple $(\pCFG,L_f,f)$ follows from the fact that $L_f(\loc,\val)+f(\val^\out) \in [0,1]$, hence condition $(C2)$ of OST-soundness is satisfied. This concludes the proof that the condition (R3) is satisfied, and the proof for $U_f$ and condition (R4) is analogous.
    \item[(R5)] Finally, by our choice of $\vecinit^1,\vecinit^2,A$, we have that $\mu^\pCFG_{\vecinit^1}(A) > e^\epsilon \mu^\pCFG_{\vecinit^2}(A)$. Hence, as $f$ is the indicator function of $f$, we also have $\E_{\val \sim \mu_{\vecinit^1}}{[f(\val^\out)]} > e^\epsilon \E_{\val \sim \mu_{\vecinit^2}}{[f(\val^\out)]}$. Therefore, since we already proved above that $L_f$ is an LESM for $f$ and that $U_f$ is a UESM for $f$, by Theorem~\ref{thm:esm-soundness} we conclude that
    \[ L_f(\locinit,\vecinit^1)+f(\vecinit^{1 \, \out}) > e^\epsilon \left( U_f(\locinit,\vecinit^2)+f(\vecinit^{2 \, \out}) \right) \]
    hence condition (R5) is satisfied.
    \end{compactenum}
    Hence, $(\vecinit^1, \vecinit^2, f, U_f, L_f)$ is indeed an expectation super/sub-martingale witness for $\epsilon$-DP refutation and this concludes our completeness proof.
\end{compactenum}
\end{proof}

\begin{remark}[Non-negativity of $f$]
    We remark that the non-negativity of $f$ in condition (R2) is necessary for our proof rule to be sound. 
    Specifically, it is used when deriving the inequality in eq.~\eqref{eq:eq1}, as otherwise $a_i$ might be negative and the inequality need not hold. 
    As a concrete example showing the necessity of (R2), consider the randomized response mechanism in Fig. \ref{fig:RR-1} with $f(\out) = -\out+\frac{1}{4}$. This function does not satisfy (R2) since $f(1) = -\frac{3}{4}$. However, $\underset{\val \sim \mu_{(0,0)}}{\E}{[f(\out)]} =0$ and $\underset{\val \sim \mu_{(1,0)}}{\E}{[f(\out)]} = - \frac{1}{2}$ which satisfies $\underset{\val \sim \mu_{(0,0)}}{\E}{[f(\out)]} > e^\epsilon \cdot \underset{\val \sim \mu_{(1,0)}}{\E}{[f(\out)]}$ for all $\epsilon \geq 0$. Hence, without (R2), one could find a witness refuting 0-DP for this example, which is incorrect because the randomized response mechanism in Fig. \ref{fig:RR-1} is known to be $(\ln 3)$-DP as shown in Example~\ref{ex:dpoverview}.
\end{remark}

\section{U/LESM and the Expectation Super/sub-martingale Witness Examples the Histogram Mechanism}\label{app:ulesmhistogram}

\begin{example}[Witness for DP refutation of the histogram mechanism]
    Consider the histogram mechanism in Fig. \ref{fig:histogram} for $n=1$ with $f(q_0) = 2q_0^2 + 393.92 q_0^4$. The following table shows an example of a UESM $U_f$ and an LESM $L_f$ for this mechanism (rounded to two decimal places):
    \begin{center}
    \begin{tabularx}{\textwidth}{c|c}
        Label & UESM \\
        \hline 
        $l_1$ & $9458.01 +  4727.0 \cdot q_0^2$ \\
        $l_2$ & $0.01 +  4.01\cdot q_0 \eta  + 1575.67 \cdot q_0 \eta^3 +2363.5\cdot q_0^2 \eta^2  + 1575.67 \cdot q_0^3 \eta  + 2.0\cdot* \eta^2  + 393.92\cdot \eta^4$ \\
        $l_t$ & $0$ \\
    \end{tabularx} \\
    \begin{tabularx}{\textwidth}{c|c}
        Label & LESM \\
        \hline 
        $l_1$ & $9458.01 +  4727.0 \cdot q_0^2$ \\
        $l_2$ & $4.01\cdot q_0 \eta  + 1575.67 \cdot q_0 \eta^3 +2363.5\cdot q_0^2 \eta^2  + 1575.67 \cdot q_0^3 \eta  + 2.0\cdot* \eta^2  + 393.92\cdot \eta^4$ \\
        $l_t$ & $0$ \\
    \end{tabularx}
    \end{center}
    It can be computed that for input $(q_{0\,\mathit{in}}^1,\eta_\mathit{in}^1) = (1.5,0)$ it holds that  $f(1.5)+U_f(l_1,(1.5,0)) \approx 22092.48$ and for the second input $(q_{0\,\mathit{in}}^2,\eta_\mathit{in}^2) = (2.5,0)$ it holds that $f(2.5)+L_f(l_1,(2.5,0)) \approx 54401.76$. Finally, we see that $54401.76 > e^{0.9} \cdot 22092.48$ which means $\big((1.5,0),(2.5,0), f, U_f, L_f\big)$ is a valid instance of an expectation super/sub-martingale witness for refuting $0.9$-DP in the histogram mechanism.
\end{example}

\section{Constraints for OST-Soundness of U/LESMs}\label{app:algoost}

In Step~2 of our algorithm in Section~\ref{sec:algo}, we are collecting constraints over the symbolic template variables that make $(\vecinit^1,\vecinit^2,f,U_f,L_f)$ a valid instance of an expectation super/sub-martingale witness for $\epsilon$-DP refutation. All the defining constraints are described in Section~\ref{sec:algo} except for the OST-soundness constraints on $U_f$ and $L_f$. As stated in Definition~\ref{def:ost}, OST-soundness can be imposed by any of the conditions (C1)-(C4). To this end, we assume an additional user provided argument that chooses one of the OST-soundness conditions to be imposed by our method. Below we recall from \cite{ChatterjeeGNZ24} how each one of (C1-4) conditions are imposed on $U_f$. The discussion is analogous for $L_f$. 

\begin{itemize}
    \item[(C1)] This condition can be applied when the pCFG $\pcfg$ has finite termination time, i.e. there exists $c>0$ such that $\timeterm_\rho \leq c$ for all runs $\rho$ of $\pcfg$. To impose this condition, we can restrict our attention to programs where all loops are statically bounded. There are no additional constraints imposed on $U_f$. 
    \item[(C2)] This condition requires the existence of a constant $c>0$ such that $|U_f(\loc,\val)+f(\val^\out)| \leq c$ for all reachable configurations $(\loc,\val)$. To impose this condition, we introduce a new unknown variable $c$, collect the constraint $c>0$ and the following constraint for every $\loc \in \locs$:
    \[
    \forall \val \in \R^\pvars. \, \val \models I(\loc) \Rightarrow |U_f(\loc,\val) + f(\val^\out)| \leq c.
    \]
    \item[(C3)] This condition has two requirements: (i) $\pcfg$ should have bounded expected termination time, i.e. $\E[\timeterm] <\infty$, and (ii) the sum $U_f(\loc,\val)+f(\val^\out)$ should have bounded expected one-step differences. The former can be verified by synthesizing a ranking supermartingale (RSM) \cite{ChakarovS13,ChatterjeeFG16,ChatterjeeFNH18} for $\pcfg$ prior to running our algorithm. To impose the latter, we introduce a new unknown variable $c$, add $c>0$ to the collected constraints and collect the following additional constraints for every $\tau =(\loc,\Pr) \in \transitions$ and $\loc' \in \support(\Pr)$:
    \[
        \forall \val \in \R^\pvars. \, \val' \in \support(\updates(\tau)), \val \vDash I(\loc) \wedge \guards(\tau) \Rightarrow |U_f(\loc,\val) + f(\val^\out) - U_f(\loc',\val') - f(\val'^\out)| \leq c.
    \]
    Note that this imposes a stronger condition than (C3), this is because our goal is to collect constraints in the form of $[\forall \val \in \R^n. \, P(\val) \Rightarrow Q(\val)]$ where $P$ is a conjunction of linear inequalities over $\pvars$ and $Q$ is a polynomial inequality over $\pvars$, and (C3) cannot directly be imposed by this kind of constraints. 
    \item[(C4)] This condition requires the existence of real numbers $M,c_1,c_2,N$ such that (i) for all sufficiently large $n \in \N$, it holds $\mathbb{P}[\timeterm >n] \leq c_1 \cdot e^{-c_2\cdot n}$, and (ii) the one step difference in $U_f(\loc,\val)+f(\val^\out)$ is bounded by $M \cdot n^d$. It was shown in \cite{Wang0GCQS19} that this condition can be verified by synthesizing an RSM similar to (C3). 
\end{itemize}

\section{Example Run of Our Algorithm: RR-1}
\begin{example}[Step 1: Setting up Templates]
        Consider the randomized response mechanism in Fig. \ref{fig:RR-1} and suppose $D=2$. In its first step, our approach fixes templates for $\vecinit^1$ as $(x_{in}^1,\out_{in}^1)$, for $\vecinit^2$ as $(x_{in}^2,\out_{in}^2)$, and the following for $f$:
        \[
        f(\out) = c_0 + c_1 \cdot \out + c_2 \cdot \out^2.
        \]
        Furthermore, the templates for $U_f$ and $L_f$ are as follows:
        \begin{align*}
            U_f(\loc,x,\out) = \begin{cases}
                             s_{\loc,0} + s_{\loc,1} \cdot x + s_{\loc,2} \cdot x^2 + s_{\loc,3} \cdot \out \cdot x + s_{\loc,4} \cdot \out + s_{\loc,5} \cdot \out^2 & \loc \neq \locterm\\ 
                            0 & \loc = \locterm
                            \end{cases} \\
            L_f(\loc,x,\out) = \begin{cases}
                             t_{\loc,0} + t_{\loc,1} \cdot x + t_{\loc,2} \cdot x^2 + t_{\loc,3} \cdot \out \cdot x + t_{\loc,4} \cdot \out + t_{\loc,5} \cdot \out^2 & \loc \neq \locterm\\ 
                            0 & \loc = \locterm
                            \end{cases}
        \end{align*}
\end{example}
\begin{example}[Step 2: Constraint Collection]
    Let us assume that after running an invariant generator, the invariant $0 \leq x,\out \leq 1$ is generated for every location of the pCFG, i.e. that $I(\loc) \equiv [0 \leq x,\out \leq 1]$ for all $\loc \in \locs$. After fixing the symbolic templates for $f, U_f$ and $L_f$, the algorithm collects the following constraints as part of its second step: 
    \[\begin{split}
        &(R1)~~~~~\out_{in}^1 = \out_{in}^2 = 0 \wedge |x_{in}^1 - x_{in}^2| \leq 1 \\
        &(R2)~~~~~[\forall \out \in \R. \, 0 \leq \out \leq 1 \Rightarrow c_0 + c_1 \cdot \out + c_2 \cdot \out^2 \geq 0]\\ 
        &(R3)~~~~~\big[\forall x, \out \in \R. \, 0 \leq x,\out \leq 1 \Rightarrow   \\
        & \qquad \qquad \qquad  L_f(l_1,x,\out) + f(\out) \leq \frac{1}{2}\big(L_f(l_2,x,out) + f(out)\big) + \frac{1}{2}\big(L_f(l_3,x,out) + f(out)\big) \big]\\
        &(R4)~~~~~\big[\forall x, \out \in \R. \, 0 \leq x,\out \leq 1 \Rightarrow U_f(l_4,x,\out) + f(\out) \leq U_f(l_t,x,0) + f(0)\big]\\
        &(R5)~~~~~L_f(\locinit,x_{in}^1,\out_{in}^1) + f(\out_{in}^1) > e^\epsilon \left( U_f(\locinit,x_{in}^2,\out_{in}^2) + f(\out_{in}^2) \right) 
    \end{split}\]
    where all occurrences of $f, U_f$ and $L_f$ are replaced by their corresponding symbolic templates.
\end{example}

\section{Proof of Theorem \ref{thm:algo-sound-complete}}\label{app:algoproof}

\begin{proof} We prove each part of the theorem separately: 
    \begin{compactitem}
        \item[\textit{Soundness.}] Suppose our algorithm terminates and returns a tuple $(\vecinit^1,\vecinit^2,f,U_f,L_f)$. We show that this tuple is indeed a valid instance of an expectation super/sub-martingale witness for $\epsilon$-DP refutation that satisfies conditions (R1)-(R5) in Definition~\ref{def:proofrule}. Given the correctness of the constraint simplification in Step 3 \cite{AsadiC0GM21}, the returned tuple is a solution of the set of constraints collected in Step 2. This means that (R1)-(R5) are all satisfied due to the collected constraints in Step 2 of our algorithm. Specifically, the OST-soundness of $U_f$ and $L_f$ is followed by the correctness of the algorithm in \cite{ChatterjeeGNZ24}, since we directly use their method for encoding OST-soundness. 
        \item[\textit{Conditional semi-completeness.}] The system of constraints collected up to Step 2 of the algorithm is satisfiable if and only if a tuple $(\vecinit^1,\vecinit^2,f,U_f,L_f)$ exists that satisfies the conditions in Definition~\ref{def:proofrule}. Given that the invariant $I$ defines a closed and bounded set, it follows that the translation via Handelman's theorem in Step 3 is sound and complete for a fixed polynomial degree $D$~\cite{AsadiC0GM21,ChatterjeeFG16}. Finally, the constraint solving of Step 4 is sound and complete, implying the conditional semi-completeness claim.   
        \item[\textit{Complexity.}] Suppose the parameter $D$ is fixed. The first step requires linear time over the size of the pCFG, since a polynomial template of degree $D$ is fixed for every location of the pCFG. The number of constraints collected in Step 2 is linear in the size of the similarity relation and the size of the pCFG. Each of these constraints is of polynomial size in the size of the pCFG. As shown in \cite{ChatterjeeGNZ24}, the translation via Handelman's theorem in Step 3 increases the size of the constraint system by a polynomial factor. Lastly, in Step 4, a system of non-linear inequalities needs to be solved in real arithmetic (an instance of existential theory of reals), which lies in PSPACE. This concludes that the total complexity of our algorithm is in PSPACE as claimed in the theorem statement.
    \end{compactitem}
\end{proof}

\pagebreak
\section{Details of Benchmarks in Section \ref{sec:experiments}} \label{app:benchmarks}
\begin{figure}[h!]
	\centering
	\begin{subfigure}[h]{0.47\textwidth}
		\begin{lstlisting}[frame=single,numbers=none]
Input: $x \in \{0,1\}$
Sim: $out^1 = out^2 = 0 \wedge |x^1 - x^2|\leq 1$
$l_1$:  if prob(0.6):
$l_2$:     $out$ := $x$
   else:
$l_3$:     $out$ := $1-x$
$l_t$:
Output: $out$
		\end{lstlisting}
        \caption{A variant of the randomized response mechanism (\texttt{RR-2}) which satisfies $(\ln 1.5)$.}
        \label{fig:RR-2}
	\end{subfigure}%
\hfill
\begin{subfigure}[h]{0.47\textwidth}
    \begin{lstlisting}[frame=single,numbers=none]
Input: $x \in [0,1]$
Sim: $out^1 = out^2 = 0 \wedge |x^1 - x^2|\leq 1$
$l_1$:  if prob($x$):
$l_2$:     $out$ := 1
   else:
$l_3$:     $out$ := 0
$l_t$:
Output: $out$
        \end{lstlisting}
    \caption{The \texttt{PrivBernoulli-1} benchmark. Provides no DP guarantees. By replacing the $[0,1]$ bound of the input with $[\frac{1}{3}, \frac{2}{3}]$, the \texttt{PrivBernoulli-2} benchmark is constructed which satisfies $(\ln 2)$-DP.}
    \label{fig:privbernoulli}
	\end{subfigure}%
    \caption{Details of \texttt{RR-2} and \texttt{PrivBernoulli} benchmarks of Section~\ref{sec:experiments}}
\end{figure}
\begin{figure}[h]
\begin{subfigure}[t]{0.47\textwidth}
		\begin{lstlisting}[frame=single,numbers=none]
Input: $x \in [0,1]$
Sim: $out^1 = out^2 = 0 \wedge |x^1 - x^2|\leq 1$
$l_1$:  if $x$==$1$:
$l_2$:     if prob($1-10^{-6}$):
$l_3$:        $out$ := 0
      else:
$l_4$:        $out$ := 0
   else: 
$l_5$:     $out$ := 0
$l_t$:
Output: $out$
            \end{lstlisting}
        \caption{The \texttt{lowprob} benchmark. Provides no DP guarantees. Creates a delusion of full differential privacy (0-DP) by almost always returning 0.}
        \label{fig:lowprob}
	\end{subfigure}%
\hfill
\begin{subfigure}[t]{0.47\textwidth}
		\begin{lstlisting}[frame=single,numbers=none]
Input: $q \in \Z$
Sim: $out^1 = out^2 = 0 \wedge |q^1 - q^2|\leq 1$
$l_1$:  $z$ := 0
$l_2$:  $\mathit{sign}$ := flip(0.5)
$l_3$:  while True:
$l_4$:     if prob(0.5):
$l_5$:        $z$ := $z+1$
      else:
$l_6$:        break
$l_7$: $out$ := $out + (-1)^\mathit{sign}\cdot z$
$l_t$:
Output: $out$
            \end{lstlisting}
        \caption{The \texttt{geometric} mechanism which satisfies $\ln 2$-DP.}
        \label{fig:geometric}
	\end{subfigure}%
        \caption{Details of \texttt{lowprob} and \texttt{geometric} benchmarks of Section~\ref{sec:experiments}}
	\label{fig:benchmarks}
\end{figure}


\section{Other Related Works} \label{app:related-works}

\subsection{Verification of DP}
Formal methods for verification of differential privacy have been studied since its introduction in the seminal papers of Dwork et al. \cite{Dwork06,DworkMNS06}. Several works considered questions related to the computational complexity of verifying DP for various classes of programs, e.g. programs with boolean variables \cite{BunGG22}, loop-free programs \cite{GaboardiNP20}, and programs with finite inputs/outputs \cite{BartheCJS020}. Logical frameworks such as probabilistic Hoare logic (pRHL) \cite{BartheGB09} were extended to support verification of DP specifications. Specifically, Barthe et al. \cite{BartheKOB13} introduced approximate pRHL (apRHL) in their machine-checked framework CertiPriv built on top of the Coq theorem prover~\cite{coq}. It was further extended to apRHL$^+$ \cite{BartheGGHS16}, which uses probabilistic lifting and couplings derived by apRHL$^+$ to prove DP. In \cite{BartheGAHKS14}, authors present a method that transforms a given probabilistic program into a non-probabilistic program simulating two executions over similar inputs together with a Hoare specification whose satisfaction implies DP of the input program. These frameworks are all capable of verifying $(\epsilon, \delta)$-DP, however, there are two main issues: (i) none of them are fully automated methods, i.e. they all require human feedback or proof fragments, and (ii) they all require the runs of the input program to be "synchronized", in the sense that runs corresponding to two similar inputs should take the exact same branches (they cannot relate two program runs with different number of random assignments). Recently, Avanzini et al. \cite{AvanziniBDG25} introduced "expectation based relational Hoare logic" (eRHL) that relaxes the synchronization requirement and provides soundness and completeness results for verification of DP. Lastly, in \cite{FarinaCG21}, authors combine the ideas of approximate probabilistic coupling and symbolic execution to prove and refute DP.

In a separate line of work, considering pure $\epsilon$-DP, LightDP \cite{ZhangK17} introduces an imperative language with a relational type system for writing provably privacy-preserving programs. The technical idea behind LightDP was to construct a "randomness alignment" between the random variables of any two neighboring runs so that they would return the same output. LightDP was presented as a step towards automation of DP verification, since it requires less manual effort compared to the logical frameworks mentioned earlier. However, it was still not fully automated and required the synchronization constraint mentioned before. Later, ShadowDP \cite{WangDWKZ19} relaxed the synchronization constraint by considering shadow executions that might take different branches. Lastly, CheckDP \cite{WangDKZ20} uses a template-based method together with a verify-invalidate loop for finding such alignments and provides automation as well. 

\end{document}